\documentclass[12pt]{amsart}
\usepackage[a4paper,margin=2.5cm]{geometry}
\usepackage[T1]{fontenc}
\usepackage[english]{babel}
\usepackage{microtype}
\usepackage{comment}
\usepackage{amsmath,amssymb,amsthm,amscd,mathtools}
\usepackage{mathrsfs}
\usepackage{enumitem}
\usepackage{xcolor}
\usepackage{xcolor}
\definecolor{citeblue}{RGB}{0,70,140}
\usepackage[
  colorlinks=true,
  linkcolor=citeblue,
  citecolor=citeblue,
  urlcolor=citeblue
]{hyperref}
\usepackage{etoolbox}
\makeatletter
\patchcmd{\@setaddresses}
  {\par\addvspace\bigskipamount\indent}
  {\par\addvspace\bigskipamount\noindent}
  {}{}
\makeatother
\makeatletter
\def\l@subsection{\@tocline{2}{0pt}{3pc}{6pc}{}}
\makeatother
\makeatletter
\renewcommand\section{\@startsection{section}{1}%
  \z@{.7\linespacing\@plus\linespacing}{.5\linespacing}%
  {\normalfont\bfseries\centering}}

\renewcommand\subsection{\@startsection{subsection}{2}%
  \z@{.5\linespacing\@plus.7\linespacing}{-.5em}%
  {\normalfont\bfseries}}
\makeatother

\numberwithin{equation}{section}

\newtheorem{theorem}{Theorem}[section]
\newtheorem{prop}[theorem]{Proposition}
\newtheorem{lemma}[theorem]{Lemma}
\newtheorem{corollary}[theorem]{Corollary}
\newtheorem{definition}[theorem]{Definition}
\newtheorem*{remark}{Remark}

\DeclareMathOperator{\Aut}{Aut}
\DeclareMathOperator{\Hom}{Hom}
\DeclareMathOperator{\Map}{Map}
\DeclareMathOperator{\Tr}{Tr}

\DeclareMathOperator{\Tors}{Tors}

\DeclareMathOperator{\Det}{det}

\title[Toral Chern--Simons TQFT]{Toral Chern--Simons TQFT\\via Geometric Quantization in Real Polarization}
\author{Daniel Galviz}
\address{YAU MATHEMATICAL SCIENCES CENTER AND DEPARTMENT OF MATHEMATICS, TSINGHUA
UNIVERSITY, BEIJING, CHINA.}
\date{}

\begin{document}

\begin{abstract}
We construct toral Chern--Simons theory with gauge group $\mathbb T=\mathfrak t/\Lambda\cong U(1)^n$ from an even, integral, nondegenerate symmetric bilinear form $K:\Lambda\times\Lambda\to\mathbb Z$ by geometric quantization via real polarization. We obtain a unitary extended $(2+1)$-dimensional TQFT by constructing the boundary state spaces and canonical operators and proving that they satisfy the cylinder and gluing axioms. The finite discriminant group $G_K=\Lambda^*/K\Lambda$ arises naturally in the theory and controls the genus-$g$ state spaces. At genus one, the theory recovers the finite quadratic data underlying bosonic Abelian topological order.
\vspace{0.5cm}
\medskip

\textbf{Main Theorem.}
Let \(\mathbb T=\mathfrak t/\Lambda\cong U(1)^n\) be a compact torus, and let
\(K:\Lambda\times\Lambda\to\mathbb Z\) be an even, integral, nondegenerate symmetric bilinear form.
Then the geometric quantization via real polarization of the moduli spaces of flat \(\mathbb T\)-connections
defines a unitary extended \((2+1)\)-dimensional TQFT
\[
Z^{\mathrm{CS}}_{\mathbb T,K}\colon \mathrm{Cob}^{\mathrm{ext}}_{2+1}\to \mathrm{Vect}_{\mathbb C}.
\]
Moreover, if \(\Sigma_g\) is a connected closed oriented surface of genus \(g\), then the vector space assigned to \(\Sigma_g\) has dimension $|G_K|^g = |\det K|^g.$
\end{abstract}

\maketitle
{\scriptsize
\setlength{\parskip}{0pt}
\hypersetup{linkcolor=black}
\tableofcontents
}

\newpage
%----------------------
\section{Introduction}
%----------------------

Recently, we showed that \(U(1)\) Chern--Simons theory and the Reshetikhin--Turaev theory
associated with the pointed modular category \(C(\mathbb Z_k,q_k)\) are naturally
isomorphic as extended \((2+1)\)-dimensional TQFTs~\cite{Galviz1}. This motivates the
geometric study of \(U(1)^n\) Chern--Simons theory as a concrete extended
\((2+1)\)-dimensional TQFT constructed by geometric quantization in real polarization.

For a closed oriented surface \(\Sigma\), the classical phase space of Abelian
Chern--Simons theory with torus gauge group \(\mathbb T\) is the moduli space of flat
\(\mathbb T\)-connections on \(\Sigma\), which is a symplectic torus. In the rank-one
case, this point of view already appears in Manoliu's construction of  $U(1)$ Abelian
Chern--Simons theory as a \((2+1)\)-dimensional TQFT \cite{Manoliu1,Manoliu2}. The toral
case requires a systematic extension of that picture to higher rank.

Although Freed, Hopkins, Lurie, and Teleman \cite{FreedHopkinsLurieTeleman} construct a fully extended \(U(1)^n\) Chern--Simons theory, the higher-categorical framework developed there is not the one needed in the present work. Their construction establishes the theory at the level of extended functorial field theory, whereas here we seek an explicit geometric model for the boundary phase spaces, state spaces, and bordism operators. Such a model is desirable both for concrete calculations and for clarifying how the lattice data, geometric quantization, and bordism operations enter the toral theory.

A complementary approach to toral Chern--Simons quantization is provided by Belov--Moore \cite{BelovMoore}, where the theory is studied in a K\"ahler framework. There one uses a complex structure on the surface to induce a complex structure on the moduli space of flat \(\mathbb T\)-connections, and the resulting quantum space is described in terms of holomorphic sections. In the present paper we instead work in real polarization. This is the natural language for our purposes, since the relevant extended topological structures are organized by Lagrangian boundary data rather than by auxiliary choices of complex structure. In particular, the real-polarization formalism is adapted to boundary Lagrangian data and makes the cylinder, gluing, and Maslov-correction structures directly visible. The Kähler and real-polarization pictures are thus complementary: the former emphasizes holomorphic geometry, while the latter is especially well suited to the extended bordism formalism developed here.

The geometric quantization of Abelian Chern--Simons theory is motivated by both physics and mathematics. On the physical side, Wen and Zee \cite{wenzee1992,Wen2016} showed that multicomponent Abelian Chern--Simons theories arise as the universal low-energy effective field theories of Abelian quantum Hall states, with the topological order encoded by the integral \(K\)-matrix. In this picture, the relevant observables are not only closed-manifold partition functions, but also the protected boundary state spaces, the operators assigned to bordisms, and the corresponding gluing laws. Our toral construction gives a geometric realization of these structures: it recovers the same topological data from geometric quantization and organizes it into an extended TQFT, thereby providing a mathematical model for the topological structure of Abelian quantum Hall phases. On the mathematical side, the same lattice-theoretic and quadratic structures appear in the theory of discriminant forms and linking pairings, as developed, for example, by Nikulin \cite{Nikulin1980} and Deloup \cite{Deloup1999}. From this point of view, the integral form \(K\) serves as a concrete lattice representative of the finite quadratic data encoded by the discriminant group \(G_K=\Lambda^*/K\Lambda\). One of the main points of the present paper is that these finite quadratic data arise naturally from geometric quantization itself.

More concretely, let \(\mathbb T=\mathfrak t/\Lambda\) be a compact torus and let
\(K:\Lambda\times\Lambda\to\mathbb Z\) be an even, integral, nondegenerate symmetric
bilinear form. We construct the symplectic boundary phase spaces
\(\mathcal M_\Sigma(\mathbb T)\) and the canonical prequantum line bundles
\(L_{\Sigma,K}\). After choosing a rational Lagrangian, we quantize in real polarization,
identify the Bohr--Sommerfeld leaves, and obtain finite-dimensional Hilbert spaces
together with canonical BKS operators whose projective defect is measured by the
Maslov--Kashiwara index. We then construct torsion-sector Chern--Simons sections and
canonical half-densities, assemble them into bordism vectors, and prove the cylinder and
gluing laws. This yields a unitary extended \((2+1)\)-dimensional TQFT in the sense of Turaev and Walker. In a companion paper, we develop a rigorous functional-integral construction of toral Chern–Simons theory \cite{Galviz3.5} and show that it produces the same extended TQFT. Together, these two constructions give a strong consistency check on the theory.

The construction proceeds in four steps: we first identify the boundary phase spaces and
their symplectic forms; next construct the prequantum line bundles and the Lagrangian
boundary data associated to \(3\)-manifolds; then quantize in real polarization and
construct the BKS operators; and finally build the bordism states and prove the TQFT
axioms.

Several ingredients are toral extensions of results already present in the rank-one
theory of Manoliu \cite{Manoliu1,Manoliu2}. In those cases we indicate precisely which
arguments are symplectic-torus generalizations and which require new input from the
lattice-valued level structure encoded by \(K\).\\

\textbf{Acknowledgements.} I would like to thank Nicolai Reshetikhin for many helpful
conversations.

\section{\texorpdfstring{Boundary Phase Space and Geometric Quantization}{Boundary Phase Space and Geometric Quantization}}
\label{sec:Toral-TQFT}

A useful conceptual point is that the existence  of a toral fully extended Chern--Simons theory is already guaranteed by the work of Freed, Hopkins, Lurie, and Teleman \cite{FreedHopkinsLurieTeleman}. For torus gauge groups they construct a fully extended theory, together with its $4$--dimensional anomaly theory, and explicitly distinguish this toral theory from Manoliu's circle-group construction, which they describe as a $2$--$3$ theory. Thus, the toral Chern--Simons theory is known to admit a $2$--$3$ TQFT formulation.

What is needed here, however, is an explicit geometric model of the boundary state spaces and bordism operators. The formulation of \cite{FreedHopkinsLurieTeleman} is anomalous and higher-categorical, and its lower-dimensional values are organized in categorical and Morita-theoretic terms. That framework is ideal for conceptual structure, but it does not provide the concrete geometric realization of boundary state spaces, BKS operators, and bordism vectors developed here. For this purpose one needs a toral analogue of Manoliu’s rank-one construction \cite{Manoliu1,Manoliu2}, extended to higher-rank torus gauge groups and to the full bordism formalism.

\subsection{Classical Boundary Phase Space}

Because \(\mathbb T=\mathfrak t/\Lambda\) is Abelian, a flat
\(\mathbb T\)--connection on \(\Sigma\) is determined up to gauge by its
holonomy class, hence by a class in
\(H^1(\Sigma;\mathbb T)\cong H^1(\Sigma;\mathfrak t)/H^1(\Sigma;\Lambda)\).
This quotient is therefore the natural classical boundary phase space, and we
denote it by \(\mathcal M_\Sigma(\mathbb T)\). The bilinear form \(K:\Lambda\times\Lambda\to\mathbb Z\) induces the symplectic form
\(\omega_{\Sigma,K}\), which is the toral analogue of the usual Chern--Simons boundary symplectic form. Let \(\Sigma\) be a closed oriented surface. Set
\begin{equation}
\mathcal{M}_\Sigma(\mathbb T)
:=
H^1(\Sigma;\mathfrak t)\big/H^1(\Sigma;\Lambda).
\end{equation}
Since \(H^1(\Sigma;\mathfrak t)\cong H^1(\Sigma;\mathbb R)\otimes\mathfrak t\)
is a finite-dimensional real vector space and \(H^1(\Sigma;\Lambda)\) is a full lattice,
\(\mathcal{M}_\Sigma(\mathbb T)\) is a compact torus \cite{FreedCS,Atiyah:1990}, with symplectic form induced by

\begin{equation}
\omega_{\Sigma,K}([\alpha],[\beta])
=
\int_\Sigma K(\alpha\wedge\beta),
\qquad
[\alpha],[\beta]\in H^1(\Sigma;\mathfrak t).
\end{equation}

\begin{prop}
\label{prop:toral-phase-space}
The pair $\bigl(\mathcal{M}_\Sigma(\mathbb T),\omega_{\Sigma,K}\bigr)$ is a compact symplectic torus.
\end{prop}

\begin{proof}
The pairing \(\omega_{\Sigma,K}\) is bilinear. Since the cup product on
\(H^1(\Sigma;\mathbb R)\) is alternating and \(K\) is symmetric,
\(\omega_{\Sigma,K}\) is alternating. The cup-product pairing on
\(H^1(\Sigma;\mathbb R)\) is nondegenerate, and \(K\) is nondegenerate on
\(\mathfrak t\), hence \(\omega_{\Sigma,K}\) is nondegenerate on \(H^1(\Sigma;\mathfrak t)\). If \(u,v\in H^1(\Sigma;\Lambda)\), then \(K(u\wedge v)\in H^2(\Sigma;\mathbb Z)\), so
\[
\omega_{\Sigma,K}(u,v)=\int_\Sigma K(u\wedge v)\in\mathbb Z.
\]
Therefore \(\omega_{\Sigma,K}\) is integral on the lattice \(H^1(\Sigma;\Lambda)\),
and so descends to a symplectic form on the quotient torus.
\end{proof}

Proposition~\ref{prop:toral-phase-space} shows that
\(\mathcal M_\Sigma(\mathbb T)\) is a finite-dimensional compact symplectic manifold. This is the classical object that will later be
quantized. In particular, the pair
\((\mathcal M_\Sigma(\mathbb T),\omega_{\Sigma,K})\) is the precise toral
replacement for the rank--one phase space in Manoliu's \(U(1)\) theory \cite{Manoliu1}.

\subsection{Boundary Polarizations and Lagrangians}

Let \(X\) be a compact oriented \(3\)--manifold with boundary. Let
\begin{equation}
r_X:H^1(X;\mathfrak t)\longrightarrow H^1(\partial X;\mathfrak t)
\end{equation}
be the restriction map and define
\begin{equation}
L_X:=\operatorname{Im}(r_X)\subset H^1(\partial X;\mathfrak t).
\end{equation}

Given a compact oriented \(3\)--manifold \(X\) with boundary \(\partial X\), not
every boundary class in \(H^1(\partial X;\mathfrak t)\) extends across \(X\).
The extendable classes form the subspace
\(L_X\), where
\(r_X\) is restriction. The
next proposition shows that \(L_X\) is Lagrangian. Geometrically, this means
that the bulk manifold \(X\) singles out a maximal classical constraint on the boundary phase space. To prove maximality, we identify the symplectic orthogonal of $L_X$ with the kernel of the connecting
homomorphism in the long exact sequence of the pair $(X,\partial X)$.

\begin{prop}
\label{prop:toral-boundary-lagrangian}
The subspace \(L_X\subset H^1(\partial X;\mathfrak t)\) is Lagrangian with
respect to $\omega_{\partial X,K}$.
\end{prop}

\begin{proof}
We first show that $L_X$ is isotropic. Let $a,b\in H^1(X;\mathfrak t)$, and choose closed
$\mathfrak t$--valued $1$--forms, still denoted $a,b$, representing these classes. Then
\[
\omega_{\partial X,K}(r_X a,r_X b)
=
\int_{\partial X} K(r_X a\wedge r_X b)
=
\int_{\partial X} r_X\!\bigl(K(a\wedge b)\bigr).
\]
Since $a$ and $b$ are closed, the $\mathfrak t$--valued bilinear form $K(a\wedge b)$ is a closed
$2$--form on $X$, so by Stokes' theorem,
\[
\int_{\partial X} r_X\!\bigl(K(a\wedge b)\bigr)
=
\int_X d\,K(a\wedge b)=0.
\]
Hence $\omega_{\partial X,K}$ vanishes on $L_X=\operatorname{Im}(r_X)$, so $L_X$ is isotropic.

It remains to show that $L_X$ is maximal isotropic, equivalently that $L_X^\perp=L_X$,  with respect to $\omega_{\partial X,K}$. Consider the long exact sequence of the pair $(X,\partial X)$:
\[
H^1(X,\partial X;\mathfrak t)\xrightarrow{\,j^*\,} H^1(X;\mathfrak t)
\xrightarrow{\,r_X\,} H^1(\partial X;\mathfrak t)
\xrightarrow{\,\delta\,} H^2(X,\partial X;\mathfrak t).
\]
We claim that $L_X^\perp=\ker(\delta).$ Since exactness gives $\ker(\delta)=\operatorname{Im}(r_X)=L_X$, this proves the result. To prove the claim, let $u\in H^1(\partial X;\mathfrak t)$. We show that
\[
u\in L_X^\perp \quad\Longleftrightarrow\quad \delta u=0.
\]
First suppose that $\delta u=0$. By exactness, there exists $a\in H^1(X;\mathfrak t)$ with
$r_X a=u$. Then for every $b\in H^1(X;\mathfrak t)$,
\[
\omega_{\partial X,K}(u,r_X b)=\omega_{\partial X,K}(r_X a,r_X b)=0
\]
by the isotropic computation above. Thus $u\in L_X^\perp$. Conversely, suppose that $u\in L_X^\perp$. We must show that $\delta u=0$. By
Poincar\'e--Lefschetz duality and the nondegeneracy of $K$, there is a natural nondegenerate pairing
\[
H^2(X,\partial X;\mathfrak t)\times H^1(X;\mathfrak t)\longrightarrow \mathbb R,
\qquad
(\eta,b)\longmapsto \int_X K(\eta\wedge b).
\]
We will show that $\delta u$ pairs trivially with every class $b\in H^1(X;\mathfrak t)$, and hence
must vanish.

Let $b\in H^1(X;\mathfrak t)$ be arbitrary, represented by a closed form still denoted $b$.
Choose any representative $\tilde u\in \Omega^1(\partial X;\mathfrak t)$ of the boundary class $u$,
and extend it to a collar neighborhood of $\partial X$. Then choose a global extension
$\widetilde u\in \Omega^1(X;\mathfrak t)$ of that collar form. In the relative de Rham model,
$\delta u$ is represented by the closed relative $2$--form $d\widetilde u$, and therefore
\[
\langle \delta u,b\rangle
=
\int_X K(d\widetilde u\wedge b).
\]
Since $db=0$, Stokes' theorem gives
\[
\int_X K(d\widetilde u\wedge b)
=
\int_X d\,K(\widetilde u\wedge b)
=
\int_{\partial X} K(u\wedge r_X b)
=
\omega_{\partial X,K}(u,r_X b).
\]
By assumption $u\in L_X^\perp$, so the last expression vanishes for every
$b\in H^1(X;\mathfrak t)$. Hence
\[
\langle \delta u,b\rangle =0
\qquad\text{for all } b\in H^1(X;\mathfrak t).
\]
By the nondegeneracy of the above pairing, $\delta u=0$. Thus $u\in \ker(\delta)$.

We have therefore proved $L_X^\perp=\ker(\delta)=\operatorname{Im}(r_X)=L_X$, so $L_X$ is Lagrangian.
\end{proof}

The Lagrangian subspace \(L_X\subset H^1(\partial X;\mathfrak t)\) is the
cohomological shadow of the bulk-extension problem: it records exactly which
boundary infinitesimal fields come from the interior of \(X\). This is the
boundary polarization associated with the bordism \(X\).

\begin{remark}
To compare with the Reshetikhin--Turaev side, it is useful to identify this cohomological Lagrangian \(L_X\) with the more familiar homological Lagrangian
\(\lambda_X\subset H_1(\Sigma;\mathbb R)\). The next lemma identifies the two
descriptions under Poincar\'e duality.
\end{remark}

\begin{lemma}\label{RT/CS boundary Lagrangian correspondence}
Let $X$ be a compact oriented $3$–manifold with boundary $\Sigma = \partial X$. Define
\[
\lambda_X := \ker\!\left(H_1(\Sigma;\mathbb{R}) \longrightarrow H_1(X;\mathbb{R})\right).
\]
Then $\lambda_X$ is a Lagrangian subspace of $H_1(\Sigma;\mathbb{R})$ with respect to the intersection pairing. Moreover, under Poincaré duality $H_1(\Sigma;\mathbb{R}) \;\cong\; H^1(\Sigma;\mathbb{R}),$ the subspace $\lambda_X$ corresponds to the Lagrangian
\[
L_X^{\mathbb R} := \operatorname{Im}\!\left(H^1(X;\mathbb R) \longrightarrow H^1(\Sigma;\mathbb R)\right), \qquad L_X=L_X^{\mathbb R}\otimes\mathfrak t.
\]
that appears in the Chern–Simons boundary phase space.
\end{lemma}

\begin{proof}
Consider the long exact sequence of the pair $(X,\Sigma)$
\[
H_2(X,\Sigma;\mathbb{R}) \longrightarrow H_1(\Sigma;\mathbb{R})
\longrightarrow H_1(X;\mathbb{R}).
\]
The image of the first map is precisely $\lambda_X$. By Poincaré–Lefschetz duality, $H_2(X,\Sigma;\mathbb{R}) \cong H^1(X;\mathbb{R}),$ and the composition $H_2(X,\Sigma;\mathbb{R}) \to H_1(\Sigma;\mathbb{R})$ is dual to the restriction map $H^1(X;\mathbb{R}) \to H^1(\Sigma;\mathbb{R}).$ Under the identification $H_1(\Sigma;\mathbb{R}) \cong H^1(\Sigma;\mathbb{R})$ given by the intersection pairing on $\Sigma$, the subspace $\lambda_X$ therefore corresponds to
\[
L_X^{\mathbb R} := \operatorname{Im}\!\left(H^1(X;\mathbb R) \longrightarrow H^1(\Sigma;\mathbb R)\right), \qquad L_X=L_X^{\mathbb R}\otimes\mathfrak t.
\]
Finally, since the corresponding subspace $L_X=L_X^{\mathbb R}\otimes \mathfrak t$ is Lagrangian by Proposition~\ref{prop:toral-boundary-lagrangian}, the subspace $\lambda_X$ is Lagrangian in $H_1(\Sigma;\mathbb{R})$ with respect to the intersection pairing.
\end{proof}

This identifies the cohomological Lagrangian \(L_X^{\mathbb R}\subset H^1(\Sigma;\mathbb R)\) selected by the bulk extension problem with the corresponding homological Lagrangian \(\lambda_X\subset H_1(\Sigma;\mathbb R)\) under Poincaré duality. Thus the boundary polarization determined by the bordism can be expressed equivalently in cohomological or homological language.

\subsection{Canonical Toral Prequantum Line Bundle}
Having identified the classical phase space
\((\mathcal M_\Sigma(\mathbb T),\omega_{\Sigma,K})\), the next step is to
construct the corresponding prequantum line bundle. For our purposes this line
bundle must do two jobs at once: it must encode the classical Chern--Simons
boundary phase, and it must carry the connection whose curvature reproduces the
symplectic form \(\omega_{\Sigma,K}\).

Let $\Sigma$ be a closed oriented surface. Let $\mathbb T=\mathfrak t/\Lambda$ be a compact
torus and let $K:\Lambda\times\Lambda\to \mathbb Z$ be an even, integral, nondegenerate symmetric bilinear form.
We write again $K$ for the induced symmetric bilinear form on $\mathfrak t$.

\begin{lemma}
\label{lem:boundary-triviality-toral}
Let $X$ be a compact oriented $3$--manifold with boundary $\Sigma$.
If $P\to X$ is a principal $\mathbb T$--bundle with characteristic class
$c(P)\in H^2(X;\Lambda)$ torsion, then the boundary bundle
$Q:=P|_\Sigma\to \Sigma$ is trivializable.
\end{lemma}

\begin{proof}
The characteristic class of $Q$ is the restriction of $c(P)$ to
$H^2(\Sigma;\Lambda)$. Since $\Sigma$ is a closed oriented surface,
\[
H^2(\Sigma;\Lambda)\cong H^2(\Sigma;\mathbb Z)\otimes \Lambda\cong \Lambda
\]
is torsion-free. Therefore the restriction of the torsion class $c(P)$ to
$H^2(\Sigma;\Lambda)$ vanishes. Hence $Q$ is topologically trivial.
\end{proof}

The point of Lemma~\ref{lem:boundary-triviality-toral} is that torsion bulk bundles become topologically trivial on the boundary surface \(\Sigma\). This
lets us work uniformly with a trivializable boundary bundle \(Q\to\Sigma\), so
that the boundary line \(L_{(Q,\eta)}\) depends only on the boundary connection
\(\eta\).

The basic idea is the same as in the Chern--Simons line construction: a boundary
field \((Q,\eta)\) does not canonically determine a vector, but any bulk filling
\((X,P,\Theta,\varphi)\) determines one, and different fillings differ by a
well-controlled \(U(1)\)--phase computed from a \(4\)--dimensional Chern--Weil term. Let us fix now a trivializable principal $\mathbb T$--bundle $Q\to\Sigma$ and let
$\eta\in \mathcal A_Q$ be a connection.

\begin{definition}
\label{def:filling-boundary-field}
A \emph{filling} of the boundary field $(Q,\eta)$ is a quadruple $(X,P,\Theta,\varphi)$,
where:
\begin{enumerate}
\item[\textup{(i)}] $X$ is a compact oriented $3$--manifold with $\partial X=\Sigma$;
\item[\textup{(ii)}] $P\to X$ is a principal $\mathbb T$--bundle;
\item[\textup{(iii)}] $\Theta\in\mathcal A_P$ is a connection on $P$;
\item[\textup{(iv)}] $\varphi:P|_\Sigma\xrightarrow{\sim}Q$ is a bundle isomorphism such that
$\varphi_*(\Theta|_\Sigma)=\eta$.
\end{enumerate}
\end{definition}

\begin{lemma}
For any two fillings of the same boundary field \((Q,\eta)\), there exists a compact
oriented \(4\)-manifold \(W\) with boundary \(X_1\cup_\Sigma(-X_2)\), together with a
principal \(\mathbb T\)-bundle \(\widetilde P\to W\) and connection \(\widetilde\Theta\)
restricting to the glued boundary data.
\end{lemma}
\begin{proof}
The existence of a comparison filling is a standard consequence of oriented bordism
with maps to \(B\mathbb T\): after gluing the two boundary fillings one obtains a class
in \(\Omega^{SO}_3(B\mathbb T)\), and since \(B\mathbb T=\mathcal K(\Lambda,2)\)\footnote{For \(\mathbb T=\mathfrak t/\Lambda\), the classifying space \(B\mathbb T\) is an Eilenberg--MacLane space \(\mathcal K(\Lambda,2)\).} has no odd
homology, the Atiyah--Hirzebruch spectral sequence implies
\(\Omega^{SO}_3(B\mathbb T)=0\); see \cite[\S 2.1]{AtiyahHirzebruch1961}.
Thus there exists a compact oriented \(4\)-manifold \(W\) with
\(\partial W=X_1\cup_\Sigma(-X_2)\) and a principal \(\mathbb T\)-bundle
\(\widetilde P\to W\) extending the glued boundary bundle. After choosing a collar of
\(\partial W\), the glued boundary connection extends over a neighborhood of
\(\partial W\), and since the space of connections on a fixed principal
\(\mathbb T\)-bundle is affine, a partition-of-unity argument extends it to a connection
\(\widetilde\Theta\) on all of \(\widetilde P\).
\end{proof}

\begin{definition}
\label{def:toral-boundary-line-filling}
Let $\mathscr F(Q,\eta)$ denote the set of fillings of $(Q,\eta)$.
Consider $[(X,P,\Theta,\varphi),z],$ for $z\in\mathbb C.$ Then, $[(X_1,P_1,\Theta_1,\varphi_1),z_1]
\sim
[(X_2,P_2,\Theta_2,\varphi_2),z_2]$
if and only if
\[
z_2
=
z_1\,
\exp\!\left(
\frac{\pi i}{(2\pi)^2}
\int_W K(F_{\widetilde\Theta}\wedge F_{\widetilde\Theta})
\right),
\]
where $W$ is any compact oriented $4$--manifold with $\partial W
=
X_1\cup_\Sigma (-X_2),$ equipped with a principal $\mathbb T$--bundle $\widetilde P\to W$ and a connection
$\widetilde\Theta$ restricting to the glued boundary field determined by
$(X_1,P_1,\Theta_1,\varphi_1)$ and $(X_2,P_2,\Theta_2,\varphi_2)$.
Define
\[
L_{(Q,\eta)}
:=
\bigl(\mathscr F(Q,\eta)\times \mathbb C\bigr)\big/\sim.
\]
\end{definition}

\begin{prop}
\label{prop:well-defined-boundary-line}
The relation of Definition~\ref{def:toral-boundary-line-filling} is well defined.
For each boundary field $(Q,\eta)$, the quotient $L_{(Q,\eta)}$ is a
one-dimensional Hermitian complex vector space.
\end{prop}

\begin{proof}
Suppose $(W,\widetilde P,\widetilde\Theta)$ and $(W',\widetilde P',\widetilde\Theta')$
are two choices used to compare the same two fillings.
Gluing $W$ to $-W'$ along their common boundary yields a closed oriented
$4$--manifold $M$ with a principal $\mathbb T$--bundle $\widehat P\to M$ and a
connection $\widehat\Theta$.
The ratio of the two comparison phases is
\[
\exp\!\left(
\frac{\pi i}{(2\pi)^2}
\int_M K(F_{\widehat\Theta}\wedge F_{\widehat\Theta})
\right).
\]
By Chern--Weil theory\footnote{We use the Chern--Weil convention \(c(\widehat P)=[F_\Theta/(2\pi)]\in H^2(X;\Lambda)\).},
\[
\frac{1}{(2\pi)^2}
\int_M K(F_{\widehat\Theta}\wedge F_{\widehat\Theta})
=
K(c(\widehat P),c(\widehat P))[M].
\]
Choose a $\Lambda$--basis so that
\[
c(\widehat P)=(c_1,\dots,c_n),\qquad K=(K_{ab})_{1\le a,b\le n}.
\]
Then
\[
K(c(\widehat P),c(\widehat P))[M]
=
\sum_a K_{aa}\,\langle c_a\smile c_a,[M]\rangle
+
2\sum_{a<b}K_{ab}\,\langle c_a\smile c_b,[M]\rangle.
\]
Because $K$ is even and integral, each diagonal coefficient $K_{aa}$ is even,
and the off-diagonal terms already carry a factor $2$. Hence this integer is even.
Therefore
\[
\exp\!\left(
\frac{\pi i}{(2\pi)^2}
\int_M K(F_{\widehat\Theta}\wedge F_{\widehat\Theta})
\right)=1.
\]
So the relation is independent of the choice of $W$ and is transitive. The quotient is a one-dimensional complex vector space because any two fillings
are related by multiplication by a unique scalar.
The norm $\bigl\|[(X,P,\Theta,\varphi),z]\bigr\|:=|z|$ is well defined because all transition factors lie in $U(1)$.
\end{proof}

Thus each boundary field \((Q,\eta)\) determines a Hermitian line \(L_{(Q,\eta)}\). The role of the comparison \(4\)--manifold \(W\) is precisely to guarantee that changing the filling changes only the phase of a vector in
that line, not the line itself.

\begin{prop}
\label{prop:gauge-action-boundary-line}
Let $u\in\mathcal G_\Sigma:=\Aut(Q)\cong \Map(\Sigma,\mathbb T)$.
For a filling $(X,P,\Theta,\varphi)$ define
\[
u\cdot (X,P,\Theta,\varphi):=(X,P,\Theta,u\circ\varphi).
\]
This induces an isometric linear map
\[
u_*:L_{(Q,\eta)}\longrightarrow L_{(Q,u\cdot \eta)}.
\]
\end{prop}

\begin{proof}
If two fillings are equivalent, then applying the same boundary gauge
transformation to both preserves the comparison $4$--manifold and the
comparison phase. Hence the map is well defined.
\end{proof}

\begin{definition}
\label{def:cylinder-transport}
Let $\eta_t$, $t\in[0,1]$, be a smooth path of connections on the fixed trivializable
bundle $Q\to \Sigma$.
Let $C=[0,1]\times \Sigma$ and let $\Xi$ be any connection on $\pi_\Sigma^*Q\to C$
restricting to $\eta_0$ and $\eta_1$ at the two boundary components.
For $[(X,P,\Theta,\varphi),z]\in L_{(Q,\eta_0)}$ define
\[
\operatorname{PT}_{\eta_0\to\eta_1}
\bigl([(X,P,\Theta,\varphi),z]\bigr)
:=
[(X\cup_\Sigma C,\;P\cup_Q \pi_\Sigma^*Q,\;\Theta\cup\Xi,\;\varphi_1),z],
\]
where $\varphi_1$ is the induced identification at $t=1$.
\end{definition}

\begin{prop}
\label{prop:boundary-line-connection-curvature}
The transport of Definition~\ref{def:cylinder-transport} is independent of the
choice of $\Xi$ and defines a natural unitary connection on the bundle of
lines $L_{(Q,\eta)}$ over $\mathcal A_Q$.
In any local trivialization of $Q$, the corresponding connection $1$--form is
\[
\alpha_\eta(\dot\eta)
=
-\pi i\int_\Sigma K(\eta\wedge \dot\eta),
\]
and its curvature is
\[
F_\alpha(\dot\eta_1,\dot\eta_2)
=
-2\pi i\int_\Sigma K(\dot\eta_1\wedge \dot\eta_2).
\]
\end{prop}

\begin{proof}
If $\Xi$ and $\Xi'$ are two cylinder connections with the same endpoint values,
glue the cylinders with opposite orientation to obtain a closed $4$--manifold.
By the argument of Proposition~\ref{prop:well-defined-boundary-line},
the resulting comparison phase is $1$. Hence the parallel transport is independent
of the choice of cylinder connection.

To compute the local connection form, fix a local trivialization of $Q$.
Then $\mathcal A_Q$ is an affine space modeled on $\Omega^1(\Sigma;\mathfrak t)$.
For a tangent vector $\dot\eta$, choose the straight-line cylinder connection $\Xi=\eta+t\dot\eta$
on $[0,\varepsilon]\times \Sigma$. Expanding the induced transport to first order
in $\varepsilon$ gives the formula for $\alpha(\dot\eta)$.
Since the tangent vectors on the affine space may be taken constant, we obtain
\[
d\alpha(\dot\eta_1,\dot\eta_2)
=
\dot\eta_1\!\left(\alpha(\dot\eta_2)\right)
-
\dot\eta_2\!\left(\alpha(\dot\eta_1)\right)
=
-2\pi i\int_\Sigma K(\dot\eta_1\wedge \dot\eta_2).
\]
This is the curvature.
\end{proof}

This is the key prequantization statement: the connection on the boundary line has curvature equal to \(-2\pi i\) times the symplectic form. In other words, the Chern--Simons line construction produces exactly the prequantum geometry
required by geometric quantization.

\begin{prop}
\label{prop:canonical-toral-prequantum}
Let $\mathcal A_Q^{\mathrm{flat}}\subset \mathcal A_Q$ be the flat locus.
Then the gauge action of Proposition~\ref{prop:gauge-action-boundary-line}
preserves the connection and descends the boundary line to a Hermitian line bundle
\[
\mathcal L_{\Sigma,K}\longrightarrow
\mathcal M_\Sigma(\mathbb T)=H^1(\Sigma;\mathfrak t)/H^1(\Sigma;\Lambda)
\]
with unitary connection $\nabla_{\Sigma,K}$ satisfying $F_{\nabla_{\Sigma,K}}=-2\pi i\,\omega_{\Sigma,K}.$ Hence $\mathcal L_{\Sigma,K}$ is the canonical prequantum line bundle of the
symplectic torus $\bigl(\mathcal M_\Sigma(\mathbb T),\omega_{\Sigma,K}\bigr)$.
\end{prop}

\begin{proof}
Gauge equivariance of the filling model gives a gauge-equivariant Hermitian bundle
over $\mathcal A_Q^{\mathrm{flat}}$.
By Proposition~\ref{prop:boundary-line-connection-curvature}, the curvature is the
pullback of the $2$--form
\[
\omega_{\Sigma,K}(\dot\eta_1,\dot\eta_2)
=
\int_\Sigma K(\dot\eta_1\wedge \dot\eta_2)
\]
on the quotient $\mathcal A_Q^{\mathrm{flat}}/\mathcal G_\Sigma
\cong \mathcal M_\Sigma(\mathbb T)$.
Hence the bundle and connection descend, with the stated curvature.
\end{proof}

We have now passed from the classical boundary data to a canonical Hermitian
line bundle \(\mathcal L_{\Sigma,K}\to \mathcal M_\Sigma(\mathbb T)\) with unitary
connection \(\nabla_{\Sigma,K}\) satisfying
\(F_{\nabla_{\Sigma,K}}=-2\pi i\,\omega_{\Sigma,K}\). This bundle
\(\mathcal L_{\Sigma,K}\) is the basic input for the quantum theory.
All subsequent quantization spaces, BKS operators, and Chern--Simons sections are
formed using $\mathcal L_{\Sigma,K}$.

\subsection{Real Polarizations and Bohr--Sommerfeld Leaves}

To quantize the symplectic torus \(\mathcal M_\Sigma(\mathbb T)\), we choose a
rational Lagrangian subspace \(L\subset H^1(\Sigma;\mathbb R)\). It determines
the translation-invariant real polarization
\(\mathcal P_L\). The
resulting leaves are compact subtori, and only some of them satisfy the
Bohr--Sommerfeld condition for the prequantum line bundle \(\mathcal L_{\Sigma,K}\). As we will show soon in Proposition~\ref{prop:toral-BS}, geometric quantization naturally produces the finite group
\[
G_K=\Lambda^*/K\Lambda,
\]
whose order is exactly the number of Bohr--Sommerfeld leaves. Thus \(G_K\) is the basic discrete invariant of our model. Let $L\subset H^1(\Sigma;\mathbb R)$ be a rational Lagrangian for the
cup-product symplectic form. Define
\[
\mathcal P_L:=L\otimes \mathfrak t\subset H^1(\Sigma;\mathfrak t).
\]
This determines a translation-invariant real polarization of
$\mathcal M_\Sigma(\mathbb T)$. Let
\[
K^\sharp:\Lambda\longrightarrow \Lambda^*:=\Hom(\Lambda,\mathbb Z),
\qquad
K^\sharp(x)(y)=K(x,y).
\]
We also write $K\Lambda:=K^\sharp(\Lambda)\subset\Lambda^*$.

\begin{prop}
\label{prop:toral-canonical-leafwise-half-density}
For each leaf $\ell$ of the polarization $\mathcal P_L$, the restriction of the half-density bundle $|\det \mathcal P_L^*|^{1/2}\big|_\ell$ carries a canonical flat connection. Moreover, the normalized
translation-invariant density of total volume $1$ on the compact torus $\ell$
admits a canonical square root
\[
\kappa_\ell^{1/2}\in
\Gamma_{\mathrm{flat}}\!\left(\ell;\,|\det \mathcal P_L^*|^{1/2}\right),
\]
which is a covariantly constant nowhere-vanishing section.
\end{prop}

\begin{proof}
The leaves of $\mathcal P_L$ are compact tori. Since $\mathcal P_L$ is
translation-invariant, it is locally spanned by translation-invariant
Hamiltonian vector fields, and the standard leafwise connection on the density
bundle is flat. Each compact torus carries a unique translation-invariant
density of total volume $1$, and this density is parallel for the leafwise flat
connection. Taking its positive square root gives the required canonical
half-density.
\end{proof}

To describe the Bohr--Sommerfeld leaves explicitly, fix a symplectic basis of
$H^1(\Sigma;\mathbb Z)$ adapted to the rational Lagrangian $L$.

\begin{prop}
\label{prop:toral-BS}
Let $\Sigma$ have genus $g$. Then the Bohr--Sommerfeld leaves of the
polarization determined by $L$ form a torsor for $G_K^g=(\Lambda^*/K\Lambda)^g.$
In particular,
\[
\#\,\mathcal{BS}(\Sigma,L)=|G_K|^g=|\det K|^g.
\]
\end{prop}

\begin{proof}
Choose a symplectic basis $(a_1,\dots,a_g,b_1,\dots,b_g)$ of $H^1(\Sigma;\mathbb Z)$ such that $$L=\operatorname{span}_{\mathbb R}\{a_1,\dots,a_g\},\qquad
L'=\operatorname{span}_{\mathbb R}\{b_1,\dots,b_g\}.$$
Since $L$ is rational and Lagrangian, after replacing the basis of $L\cap H^1(\Sigma;\mathbb Z)$ if
necessary, we may choose such a basis adapted to $L$. Using the decomposition
\[
H^1(\Sigma;\mathfrak t)\cong (L\otimes \mathfrak t)\oplus (L'\otimes \mathfrak t),
\]
every point of $\mathcal M_\Sigma(\mathbb T)=H^1(\Sigma;\mathfrak t)/H^1(\Sigma;\Lambda)$ can be represented by coordinates
\[
(x_1,\dots,x_g;y_1,\dots,y_g)\in \mathfrak t^g\times \mathfrak t^g,
\]
where the $x_i$ are coordinates along $L\otimes\mathfrak t$ and the $y_i$ are coordinates along
$L'\otimes\mathfrak t$. The polarization $\mathcal P_L=L\otimes\mathfrak t$ is tangent to the $x$--directions, so its leaves are obtained by
fixing the $y$--coordinates. Hence each leaf is of the form
\[
\ell_y=\{(x,y):x\in \mathfrak t^g\}/(L\otimes\mathfrak t\cap H^1(\Sigma;\Lambda)),
\]
and therefore is a translate of the compact torus
\[
\mathcal P_L/(\mathcal P_L\cap H^1(\Sigma;\Lambda))\cong (\mathfrak t/\Lambda)^g=\mathbb T^g.
\]

We now compute the Bohr--Sommerfeld condition. By the explicit construction of the prequantum line
bundle $L_{\Sigma,K}$, its curvature is
\(-2\pi i\,\omega_{\Sigma,K},\)
where in the above coordinates
\[
\omega_{\Sigma,K}
=
\sum_{j=1}^g K(dx_j\wedge dy_j).
\]
Accordingly, one may choose on the universal cover $H^1(\Sigma;\mathfrak t)\cong \mathfrak t^g\times\mathfrak t^g$, a translation-invariant connection $1$--form for the prequantum bundle of the form
\[
\alpha
=
\pi i\sum_{j=1}^g\bigl(K(x_j,dy_j)-K(y_j,dx_j)\bigr),
\]
whose exterior derivative is $d\alpha=-2\pi i\,\omega_{\Sigma,K}$. Restricting to the leaf $\ell_y$, the
coordinates $y_j$ are constant, so
\[
\alpha|_{\ell_y}=-\pi i\sum_{j=1}^g K(y_j,dx_j).
\]
To make the lattice automorphy factor explicit, let  $\lambda=(\lambda_1,\dots,\lambda_g)\in \Lambda^g$ and denote by \(\tau_\lambda\) the translation $\tau_\lambda(x,y)=(x+\lambda,y)$ of the universal cover \(H^1(\Sigma;\mathfrak t)\cong \mathfrak t^g\times \mathfrak t^g\).
A direct computation gives
\[
\tau_\lambda^*\alpha-\alpha
=
\pi i\sum_{j=1}^g K(\lambda_j,dy_j)
=
d\!\left(\pi i\sum_{j=1}^g K(\lambda_j,y_j)\right).
\]
Thus, with the chosen normalization of \(L_{\Sigma,K}\), translation by \(\lambda\) identifies the
fiber over \((x+\lambda,y)\) with the fiber over \((x,y)\) by multiplication by
\[
\exp\!\left(-\pi i\sum_{j=1}^g K(\lambda_j,y_j)\right).
\]

Now fix a leaf \(\ell_y\). Since the coordinates \(y_j\) are constant along \(\ell_y\), the
restriction of the connection form is
\[
\alpha|_{\ell_y}=-\pi i\sum_{j=1}^g K(y_j,dx_j).
\]
Hence a covariantly constant section on the universal cover of \(\ell_y\) has the form
\[
s_y(x)=
\exp\!\left(\pi i\sum_{j=1}^g K(y_j,x_j)\right)c,
\qquad c\in \mathbb C.
\]
For such a section one has
\[
s_y(x+\lambda)
=
\exp\!\left(\pi i\sum_{j=1}^g K(y_j,\lambda_j)\right)s_y(x).
\]
Therefore \(s_y\) descends to a well-defined flat section on the quotient leaf if and only if
\[
\exp\!\left(\pi i\sum_{j=1}^g K(y_j,\lambda_j)\right)
=
\exp\!\left(-\pi i\sum_{j=1}^g K(\lambda_j,y_j)\right)
\]
for every \(\lambda\in\Lambda^g\). Since \(K\) is symmetric, this is equivalent to
\[
\exp\!\left(2\pi i\sum_{j=1}^g K(y_j,\lambda_j)\right)=1
\qquad\text{for all }\lambda\in\Lambda^g.
\]
Because the \(\lambda_j\) are independent, this holds if and only if
\[
K(y_j,\lambda_j)\in\mathbb Z
\qquad\text{for all }\lambda_j\in\Lambda,\ \ j=1,\dots,g.
\]
Equivalently, $Ky_j\in\Lambda^*\text{ for each }j,$ or, what is the same, $y_j\in K^{-1}\Lambda^*.$ Thus the Bohr--Sommerfeld leaves are precisely those with
\[
y\in (K^{-1}\Lambda^*)^g.
\]
Two parameters $y,y'\in (K^{-1}\Lambda^*)^g$ determine the same leaf in $\mathcal M_\Sigma(\mathbb T)$ if and only if
they differ by an element of $\Lambda^g$, since the $y$--coordinates are defined modulo the lattice in the
transverse factor. Therefore the set of Bohr--Sommerfeld leaves is naturally identified with $(K^{-1}\Lambda^*/\Lambda)^g.$ Multiplication by $K$ induces an isomorphism
\[
K^{-1}\Lambda^*/\Lambda \xrightarrow{\ \cong\ } \Lambda^*/K\Lambda = G_K,
\]
so we obtain $\mathcal{BS}(\Sigma,L)\cong G_K^g.$
This identification depends on the choice of origin leaf, so canonically the Bohr--Sommerfeld set is a
torsor for $G_K^g$. Finally, since $K$ is nondegenerate, $G_K=\Lambda^*/K\Lambda$ is finite of order $|G_K|=|\det K|,$ hence
\[
\#\mathcal{BS}(\Sigma,L)=|G_K|^g=|\det K|^g.
\]
\end{proof}

The finite discriminant group $G_K=\Lambda^*/K\Lambda$ is the arithmetic datum controlling the
Bohr--Sommerfeld set. In genus $g$, the allowed leaves form a $G_K^g$--torsor, so the quantum Hilbert
space has dimension $|G_K|^g=|\det K|^g$.

\subsection{The Toral Quantum Hilbert Space}

Quantization in a real polarization means taking covariantly constant sections along the polarized leaves. Because the leaves of \(\mathcal P_L\) are compact, parallel sections can occur only on Bohr--Sommerfeld leaves, and the half-density
factor \( |\det\mathcal P_L^*|^{1/2}\) provides the correct measure-theoretic normalization.

\begin{definition}
\label{def:toral-hilbert}
Let $(\Sigma,L)$ be a closed oriented surface endowed with a rational
Lagrangian $L\subset H^1(\Sigma;\mathbb R)$. For each
Bohr--Sommerfeld leaf $\ell\in \mathcal{BS}(\Sigma,L)$ define
\[
\mathscr S_\ell:=
\Gamma_{\mathrm{flat}}
\!\left(
\ell;\,
\mathcal L_{\Sigma,K}\otimes |\det\mathcal P_L^*|^{1/2}
\right).
\]
The toral Hilbert space is
\[
\mathcal H_{\mathbb T,K}(\Sigma,L)
:=
\bigoplus_{\ell\in \mathcal{BS}(\Sigma,L)}\mathscr S_\ell.
\]
\end{definition}

\begin{prop}
\label{prop:toral-BS-one-dimensional}
Let $\ell$ be a leaf of $\mathcal P_L$.
Then
\[
\Gamma_{\mathrm{flat}}
\!\left(
\ell;\,
\mathcal L_{\Sigma,K}\otimes |\det\mathcal P_L^*|^{1/2}
\right)
=
0
\]
unless $\ell$ is Bohr--Sommerfeld. If $\ell$ is Bohr--Sommerfeld, then
\[
\dim
\Gamma_{\mathrm{flat}}
\!\left(
\ell;\,
\mathcal L_{\Sigma,K}\otimes |\det\mathcal P_L^*|^{1/2}
\right)
=1.
\]
\end{prop}

\begin{proof}
By Proposition~\ref{prop:toral-canonical-leafwise-half-density}, the
half-density factor has a canonical parallel trivialization on every leaf.
Hence the existence of parallel sections is equivalent to trivial holonomy of
the leafwise flat connection on $\mathcal L_{\Sigma,K}|_\ell$.
This is exactly the Bohr--Sommerfeld condition.
On a Bohr--Sommerfeld leaf the resulting flat line bundle has trivial holonomy,
so its space of parallel sections is one-dimensional.
\end{proof}

Each Bohr--Sommerfeld leaf contributes exactly one quantum state, and non-Bohr--Sommerfeld leaves do not contribute. Thus the Hilbert space \(\mathcal H_{\mathbb T,K}(\Sigma,L)\) is  the direct sum of the one-dimensional contributions from the allowed leaves.

\begin{theorem}
\label{thm:toral-hilbert-dimension}
If $\Sigma$ has genus $g$, then
\[
\dim \mathcal H_{\mathbb T,K}(\Sigma,L)
=
|G_K|^g
=
|\det K|^g.
\]
\end{theorem}

\begin{proof}
By Proposition~\ref{prop:toral-BS-one-dimensional}, each Bohr--Sommerfeld leaf
contributes a one-dimensional summand. By
Proposition~\ref{prop:toral-BS}, the number of such leaves is
$|G_K|^g$.
\end{proof}

The dimension formula \(\dim \mathcal H_{\mathbb T,K}(\Sigma,L)=|\det K|^g\) is the first important
numerical invariant of the toral theory. It will reappear later in the
normalization of the cylinder and in the gluing formula.

\subsection{BKS Operators for Toral Real Polarizations}
The Hilbert space \(\mathcal H_{\mathbb T,K}(\Sigma,L)\) depends a priori on the choice of rational Lagrangian \(L\subset H^1(\Sigma;\mathbb R)\). To obtain a
TQFT, one must compare these different quantizations canonically. This is the role of the Blattner--Kostant--Sternberg operator
\(F_{L_2L_1}\). Because the moduli space $\mathcal M_\Sigma(\mathbb T)$ is a compact symplectic torus and the
polarizations $\mathcal P_L$ are translation-invariant with compact leaves, the  BKS via real polarization formalism  applies unchanged in the present setting.

\begin{lemma}
\label{lem:toral-maslov-signature}
Let $\mu_\Sigma(L_1,L_2,L_3)$ denote the Maslov--Kashiwara index of the triple
$L_1,L_2,L_3 \subset H^1(\Sigma;\mathbb R)$ for the cup-product symplectic form
$\omega_\Sigma$, and define
\[
\mu_K(L_1,L_2,L_3)
:=
\mu(\mathcal P_{L_1},\mathcal P_{L_2},\mathcal P_{L_3}),
\qquad
\mathcal P_{L_i}=L_i\otimes \mathfrak t
\subset H^1(\Sigma;\mathfrak t).
\]
Then
\[
\mu_K(L_1,L_2,L_3)=\sigma(K)\,\mu_\Sigma(L_1,L_2,L_3),
\]
where $\sigma(K)$ is the signature of the real bilinear form $K$ on $\mathfrak t$.
\end{lemma}

\begin{proof}
Choose a basis $(e_1,\dots,e_r)$ of $\mathfrak t$ in which the real symmetric form $K$
is diagonal:
\[
K(e_a,e_b)=\varepsilon_a\delta_{ab},
\qquad
\varepsilon_a\in\{+1,-1\}.
\]
Then
\[
H^1(\Sigma;\mathfrak t)
\cong
\bigoplus_{a=1}^r H^1(\Sigma;\mathbb R)e_a,
\qquad
\omega_{\Sigma,K}
=
\bigoplus_{a=1}^r \varepsilon_a\,\omega_\Sigma,
\]
and for each rational Lagrangian $L_i\subset H^1(\Sigma;\mathbb R)$,
\[
\mathcal P_{L_i}
=
L_i\otimes\mathfrak t
\cong
\bigoplus_{a=1}^r L_i e_a.
\]

Let
\[
Q_\Sigma(x_1,x_2,x_3)
=
\omega_\Sigma(x_1,x_2)+\omega_\Sigma(x_2,x_3)+\omega_\Sigma(x_3,x_1)
\]
be the quadratic form whose signature is $\mu_\Sigma(L_1,L_2,L_3)$ see \cite[Secs.~IV, VI]{Manoliu1}.
Under the above decomposition, the quadratic form defining
$\mu_K(L_1,L_2,L_3)$ is the orthogonal direct sum
\[
Q_K
\cong
\bigoplus_{a=1}^r \varepsilon_a\, Q_\Sigma.
\]
Therefore
\[
\mu_K(L_1,L_2,L_3)
=
\operatorname{sgn}(Q_K)
=
\sum_{a=1}^r \varepsilon_a\,\operatorname{sgn}(Q_\Sigma)
=
\sigma(K)\,\mu_\Sigma(L_1,L_2,L_3).
\]
\end{proof}
Lemma \ref{lem:toral-maslov-signature} is the explicit \(K\)-matrix refinement of the
standard BKS Fourier kernel for transverse invariant real polarizations on a compact
symplectic torus, proved in the rank-one case by Manoliu
\cite[Secs.~IV, VI]{Manoliu1}; compare also \cite[Sec.~V]{Manoliu2}.
\begin{prop}
\label{prop:toral-BKS}
Let $\Sigma$ be a closed oriented surface, and let
$L_1,L_2,L_3\subset H^1(\Sigma;\mathbb R)$ be rational Lagrangian subspaces.
Then there exists a canonical unitary
Blattner--Kostant--Sternberg operator
\[
F_{L_2L_1}:
\mathcal H_{\mathbb T,K}(\Sigma,L_1)\longrightarrow
\mathcal H_{\mathbb T,K}(\Sigma,L_2).
\]
Moreover, for any three rational Lagrangians $L_1,L_2,L_3$,
\begin{equation}
F_{L_3L_2}\circ F_{L_2L_1}
=
e^{\frac{\pi i}{4}\mu_K(L_1,L_2,L_3)}\,F_{L_3L_1}
=
e^{\frac{\pi i}{4}\sigma(K)\mu_\Sigma(L_1,L_2,L_3)}\,F_{L_3L_1}.
\end{equation}
\end{prop}

\begin{proof}
Set $\mathcal M_\Sigma(\mathbb T)
:=
H^1(\Sigma;\mathfrak t)/H^1(\Sigma;\Lambda),$
equipped with the translation-invariant symplectic form $\omega_{\Sigma,K}$ constructed above,
and let $L_{\Sigma,K}\to\mathcal M_\Sigma(\mathbb T)$
be the prequantum line bundle with curvature $-2\pi i\,\omega_{\Sigma,K}$.
For a rational Lagrangian subspace $L\subset H^1(\Sigma;\mathbb R)$, let
\[
\mathcal P_L:=L\otimes \mathfrak t
\subset T\mathcal M_\Sigma(\mathbb T)
\]
denote the corresponding translation-invariant real polarization.

Since $H^1(\Sigma;\mathfrak t)$ is a finite-dimensional real vector space and
$H^1(\Sigma;\Lambda)$ is a full lattice in it, the quotient
$\mathcal M_\Sigma(\mathbb T)$ is a compact torus. The form $\omega_{\Sigma,K}$ is constant and
nondegenerate, so $\mathcal M_\Sigma(\mathbb T)$ is a compact symplectic torus. By construction,
$L_{\Sigma,K}$ is a prequantum line bundle for this symplectic torus.

Now fix a rational Lagrangian $L\subset H^1(\Sigma;\mathbb R)$. Because $L$ is Lagrangian,
$\mathcal P_L=L\otimes\mathfrak t$ is an involutive isotropic distribution of half the dimension of
$\mathcal M_\Sigma(\mathbb T)$; hence it is a real polarization. Because $L$ is rational, the leaves of
$\mathcal P_L$ are compact subtori of $\mathcal M_\Sigma(\mathbb T)$, namely translates of
$\mathcal P_L/(\mathcal P_L\cap H^1(\Sigma;\Lambda))$.

By Proposition \ref{prop:toral-canonical-leafwise-half-density}, the leafwise half-density bundle along $\mathcal P_L$ carries a canonical flat
structure and a canonical normalized covariantly constant section on each leaf. Then by Definition \ref{def:toral-hilbert} and Proposition \ref{prop:toral-BS-one-dimensional},
the corresponding quantization via real polarization is
\[
\mathcal H_{\mathbb T,K}(\Sigma,L)
=
\bigoplus_{\ell\in \mathcal{BS}(\Sigma,L)}
\Gamma_{\mathrm{flat}}\!\left(
\ell;\,
L_{\Sigma,K}\otimes |\det \mathcal P_L^*|^{1/2}
\right),
\]
with each Bohr--Sommerfeld leaf contributing a one-dimensional summand. Thus Manoliu’s BKS formalism for invariant real polarizations of compact symplectic tori applies in the present setting; see \cite[ Sections II, IV]{Manoliu2}. The resulting comparison map is a canonical unitary isomorphism, and its projective composition law is given by \cite[Theorem VI.4]{Manoliu2}. Applying that construction to the pair $(\mathcal P_{L_1},\mathcal P_{L_2})$
yields a canonical unitary isomorphism
\[
F_{L_2L_1}\colon
\mathcal H_{\mathbb T,K}(\Sigma,L_1)\xrightarrow{\ \cong\ }
\mathcal H_{\mathbb T,K}(\Sigma,L_2).
\]

For three rational Lagrangians $L_1,L_2,L_3\subset H^1(\Sigma;\mathbb R)$, Manoliu's compact-torus
BKS composition law gives
\[
F_{L_3L_2}\circ F_{L_2L_1}
=
e^{\frac{\pi i}{4}\mu(\mathcal P_{L_1},\mathcal P_{L_2},\mathcal P_{L_3})}
\,F_{L_3L_1}.
\]
By definition, $\mu(\mathcal P_{L_1},\mathcal P_{L_2},\mathcal P_{L_3})
=
\mu_K(L_1,L_2,L_3),$ and by Lemma \ref{lem:toral-maslov-signature},
\[
\mu_K(L_1,L_2,L_3)=\sigma(K)\,\mu_\Sigma(L_1,L_2,L_3).
\]
Therefore
\[
F_{L_3L_2}\circ F_{L_2L_1}
=
e^{\frac{\pi i}{4}\mu_K(L_1,L_2,L_3)}\,F_{L_3L_1}
=
e^{\frac{\pi i}{4}\sigma(K)\mu_\Sigma(L_1,L_2,L_3)}\,F_{L_3L_1},
\]
as claimed.
\end{proof}

The BKS operator \(F_{L_2L_1}\) gives the canonical unitary identification between the quantizations defined by \(L_1\) and \(L_2\). Its failure to be strictly transitive is measured by the Maslov--Kashiwara index
\(\mu(L_1,L_2,L_3)\), and this is exactly the phase that $K$-twisted weighting will later cancel in the extended TQFT formalism.

\begin{lemma}\label{lem:genus-BKS-fourier}
Let $\Sigma$ be a closed oriented surface of genus $g$, and choose a symplectic basis
\[
(a_1,\dots,a_g,b_1,\dots,b_g)
\]
of $H^1(\Sigma;\mathbb Z)$. Set
\[
L=\mathbb R\langle a_1,\dots,a_g\rangle,\qquad
L'=\mathbb R\langle b_1,\dots,b_g\rangle .
\]
Identify $\mathcal{BS}(\Sigma,L)$ and $\mathcal{BS}(\Sigma,L')$ with $G_K^g$
by choosing the Bohr--Sommerfeld leaf through the origin in each polarization as the
basepoint, as in the coordinate description of Proposition~\ref{prop:toral-BS}, and let
\[
\{e_{\mathbf u}\}_{\mathbf u\in G_K^g}\subset \mathcal H_{\mathbb T,K}(\Sigma,L),\qquad
\{e_{\mathbf v}'\}_{\mathbf v\in G_K^g}\subset \mathcal H_{\mathbb T,K}(\Sigma,L')
\]
be the corresponding bases determined by the canonical half-densities of
Proposition~\ref{prop:toral-canonical-leafwise-half-density}. Then, after fixing the
phases of the vectors $e_{\mathbf v}'$, the BKS operator of
Proposition~\ref{prop:toral-BKS} is given by
\[
F_{L',L}(e_{\mathbf u})
=
|G_K|^{-g/2}\sum_{\mathbf v\in G_K^g}\Omega_{K,g}(\mathbf u,\mathbf v)\,e_{\mathbf v}',
\qquad \mathbf u\in G_K^g,
\]
where, for $\mathbf u=(u_1,\dots,u_g)$ and $\mathbf v=(v_1,\dots,v_g)$,
\[
\Omega_{K,g}(\mathbf u,\mathbf v)
:=
\prod_{j=1}^g \Omega_K(u_j,v_j)
=
\exp\!\Bigl(2\pi i\sum_{j=1}^g \tilde u_j^{\top}K^{-1}\tilde v_j\Bigr),
\]
with $u_j=[\tilde u_j]$ and $v_j=[\tilde v_j]$ in $G_K=\Lambda^*/K\Lambda$.
Equivalently, after identifying the two bases by their common labels in $G_K^g$,
\[
F_{L',L}(e_{\mathbf u})
=
|G_K|^{-g/2}\sum_{\mathbf v\in G_K^g}\Omega_{K,g}(\mathbf u,\mathbf v)\,e_{\mathbf v}.
\]
\end{lemma}

\begin{proof}
Using the coordinates of the proof of Proposition~\ref{prop:toral-BS}, we have
\[
\mathcal M_\Sigma(\mathbb T)
\cong H^1(\Sigma;\mathfrak t)/H^1(\Sigma;\Lambda)
\cong (\mathfrak t^g\oplus \mathfrak t^g)/(\Lambda^g\oplus\Lambda^g),
\]
with coordinates $(x,y)=((x_1,\dots,x_g),(y_1,\dots,y_g))$ adapted to the decomposition
\[
H^1(\Sigma;\mathfrak t)=(L\otimes \mathfrak t)\oplus (L'\otimes \mathfrak t).
\]
Thus the $L$-leaves are
\[
\ell_y=\{(x,y):x\in \mathfrak t^g\}/\Lambda^g,
\]
and the $L'$-leaves are
\[
m_x=\{(x,y):y\in \mathfrak  t^g\}/\Lambda^g.
\]

By Proposition~\ref{prop:toral-BS}, one may choose on the universal cover a connection
$1$-form
\[
\alpha
=
\pi i\sum_{j=1}^g\bigl(K(x_j,dy_j)-K(y_j,dx_j)\bigr).
\]
Hence
\[
\alpha|_{\ell_y}
=
-\pi i\sum_{j=1}^g K(y_j,dx_j),
\qquad
\alpha|_{m_x}
=
\pi i\sum_{j=1}^g K(x_j,dy_j).
\]
Therefore the leafwise flat sections are
\[
s_y(x)
=
\exp\!\Bigl(\pi i\sum_{j=1}^g K(y_j,x_j)\Bigr)\,\kappa_{\ell_y}^{1/2},
\qquad
t_x(y)
=
\exp\!\Bigl(-\pi i\sum_{j=1}^g K(x_j,y_j)\Bigr)\,\kappa_{m_x}^{1/2},
\]
where $\kappa_{\ell_y}^{1/2}$ and $\kappa_{m_x}^{1/2}$ are the canonical normalized
parallel half-densities from
Proposition~\ref{prop:toral-canonical-leafwise-half-density}. Now let $\mathbf u=(u_1,\dots,u_g)$ and $\mathbf v=(v_1,\dots,v_g)$ in $G_K^g$, and choose
lifts $\tilde u_j,\tilde v_j\in\Lambda^*$. Write
\[
y_{\mathbf u}:=(K^{-1}\tilde u_1,\dots,K^{-1}\tilde u_g),\qquad
x_{\mathbf v}:=(K^{-1}\tilde v_1,\dots,K^{-1}\tilde v_g).
\]
Then $\ell_{y_{\mathbf u}}$ and $m_{x_{\mathbf v}}$ are Bohr--Sommerfeld leaves by
Proposition~\ref{prop:toral-BS}, and they intersect in the unique point
$(x_{\mathbf v},y_{\mathbf u})$. The standard BKS pairing for transverse invariant real
polarizations on a compact symplectic torus gives
\[
\langle e_{\mathbf v}',F_{L',L}e_{\mathbf u}\rangle
=
c\,\overline{t_{x_{\mathbf v}}(y_{\mathbf u})}\,s_{y_{\mathbf u}}(x_{\mathbf v})
\]
for a constant $c$ independent of $\mathbf u,\mathbf v$. Evaluating at the intersection
point,
\[
\overline{t_{x_{\mathbf v}}(y_{\mathbf u})}\,s_{y_{\mathbf u}}(x_{\mathbf v})
=
\exp\!\Bigl(2\pi i\sum_{j=1}^g \tilde u_j^{\top}K^{-1}\tilde v_j\Bigr)
=
\Omega_{K,g}(\mathbf u,\mathbf v),
\]
since $K$ is symmetric. Thus
\[
F_{L',L}(e_{\mathbf u})
=
c\sum_{\mathbf v\in G_K^g}\Omega_{K,g}(\mathbf u,\mathbf v)\,e_{\mathbf v}'.
\]

It remains to determine $c$. By Proposition~\ref{prop:toral-BKS}, $F_{L',L}$ is unitary,
hence
\[
\delta_{\mathbf u,\mathbf u'}
=
\langle F_{L',L}e_{\mathbf u},F_{L',L}e_{\mathbf u'}\rangle
=
|c|^2\sum_{\mathbf v\in G_K^g}
\overline{\Omega_{K,g}(\mathbf u,\mathbf v)}\,\Omega_{K,g}(\mathbf u',\mathbf v).
\]
Since $\Omega_{K,g}$ is a bicharacter,
\[
\overline{\Omega_{K,g}(\mathbf u,\mathbf v)}\,\Omega_{K,g}(\mathbf u',\mathbf v)
=
\Omega_{K,g}(\mathbf u'-\mathbf u,\mathbf v).
\]
If $\mathbf u'=\mathbf u$, the sum is $|G_K|^g$. If $\mathbf u'\neq \mathbf u$, then
$\mathbf v\mapsto \Omega_{K,g}(\mathbf u'-\mathbf u,\mathbf v)$ is a nontrivial character
of the finite abelian group $G_K^g$, so its sum over $G_K^g$ is $0$. Therefore
\[
\sum_{\mathbf v\in G_K^g}
\overline{\Omega_{K,g}(\mathbf u,\mathbf v)}\,\Omega_{K,g}(\mathbf u',\mathbf v)
=
|G_K|^g\,\delta_{\mathbf u,\mathbf u'}.
\]
It follows that $|c|^2|G_K|^g=1$, hence $|c|=|G_K|^{-g/2}$. After rephasing the basis
$\{e_{\mathbf v}'\}_{\mathbf v\in G_K^g}$, we may take $c=|G_K|^{-g/2}$.
\end{proof}

\begin{lemma}\label{lem:genus-T-diagonal}
Let $\Sigma$ be a closed oriented surface of genus $g$, choose a symplectic basis $(a_1,\dots,a_g,b_1,\dots,b_g)$ of $H^1(\Sigma;\mathbb Z)$, and set $L=\mathbb R\langle a_1,\dots,a_g\rangle.$ For each $1\le j\le g$, let $\tau_j\in \mathrm{Map}^+(\Sigma)$ be the Dehn twist about the
$a_j$-cycle. In the coordinates
\[
(x_1,\dots,x_g,y_1,\dots,y_g)
\]
of Proposition~\ref{prop:toral-BS}, the mapping class $\tau_j$ acts by
\[
\tau_j(x_1,\dots,x_g,y_1,\dots,y_g)
=
(x_1,\dots,x_j+y_j,\dots,x_g,y_1,\dots,y_g).
\]
Since the polarization $L$ has leaves
\[
\ell_y=\{(x,y):x\in \mathfrak  t^g\}/\Lambda^g,
\]
the map $\tau_j$ preserves each leaf $\ell_y$. Hence the induced operator
\[
T_j:\mathcal H_{\mathbb T,K}(\Sigma,L)\to \mathcal H_{\mathbb T,K}(\Sigma,L)
\]
acts diagonally in the Bohr--Sommerfeld basis
$\{e_{\mathbf u}\}_{\mathbf u\in G_K^g}$. More precisely, if
\[
\mathbf u=(u_1,\dots,u_g)\in G_K^g,
\]
then
\[
T_j(e_{\mathbf u})=q_K(u_j)\,e_{\mathbf u}.
\]
In particular, if $\tau_A:=\tau_1\cdots \tau_g$, then
\[
T_A(e_{\mathbf u})
=
\Bigl(\prod_{j=1}^g q_K(u_j)\Bigr)e_{\mathbf u}.
\]
\end{lemma}

\begin{proof}
For a leaf $\ell_y$ of the polarization $P_L$, the proof of Proposition~\ref{prop:toral-BS}
gives
\[
\alpha|_{\ell_y}
=
-\pi i\sum_{i=1}^g K(y_i,dx_i),
\]
hence a normalized leafwise flat section
\[
s_y(x)
=
\exp\!\Bigl(\pi i\sum_{i=1}^g K(y_i,x_i)\Bigr)\,\kappa_{\ell_y}^{1/2}.
\]
Since $\tau_j$ preserves the polarization $L$, it preserves each leaf $\ell_y$ and acts on
it by the translation
\[
(x_1,\dots,x_g)\longmapsto (x_1,\dots,x_j+y_j,\dots,x_g).
\]
Moreover, $\tau_j$ preserves the normalized translation-invariant half-density
$\kappa_{\ell_y}^{1/2}$. Therefore
\[
(\tau_j^*s_y)(x)
=
s_y(x_1,\dots,x_j+y_j,\dots,x_g)
=
\exp\!\bigl(\pi i\,K(y_j,y_j)\bigr)\,s_y(x).
\]
Now let $\mathbf u=(u_1,\dots,u_g)\in G_K^g$, choose lifts $\tilde u_j\in \Lambda^*$, and
write
\[
y_{\mathbf u}=(K^{-1}\tilde u_1,\dots,K^{-1}\tilde u_g).
\]
Then $e_{\mathbf u}$ is the basis vector corresponding to the Bohr--Sommerfeld leaf
$\ell_{y_{\mathbf u}}$, and
\[
\exp\!\bigl(\pi i\,K((y_{\mathbf u})_j,(y_{\mathbf u})_j)\bigr)
=
\exp\!\bigl(\pi i\,\tilde u_j^{\top}K^{-1}\tilde u_j\bigr)
=
q_K(u_j).
\]
Hence
\[
T_j(e_{\mathbf u})=q_K(u_j)\,e_{\mathbf u}.
\]
The formula for $T_A$ follows because the twists $\tau_1,\dots,\tau_g$ commute.
\end{proof}

\section{Torsion Sectors and Canonical Boundary States}
\subsection{Moduli of Flat Toral Bundles}

For a compact oriented \(3\)--manifold \(X\), the classical Abelian theory is
naturally decomposed according to the topological type of the underlying \(\mathbb T\)--bundle. In the toral case, the relevant topological data are the
characteristic class in \(H^2(X;\Lambda)\), and flat connections occur exactly on the torsion components. The short exact sequence $0\longrightarrow \Lambda\longrightarrow \mathfrak t\longrightarrow \mathbb T\longrightarrow 0$ induces a long exact sequence in cohomology. Since $\mathfrak t$ is a real
vector space, $H^2(X;\mathfrak t)$ is torsion-free.

\begin{prop}
\label{prop:toral-flat-bundles-rigorous}
Let $X$ be a compact oriented $3$--manifold.

\begin{enumerate}
\item[\textup{(i)}]
Isomorphism classes of principal $\mathbb T$--bundles over $X$ are classified by
$H^2(X;\Lambda)$.

\item[\textup{(ii)}]
A principal $\mathbb T$--bundle admits a flat connection if and only if its
characteristic class lies in $\Tors H^2(X;\Lambda)$.

\item[\textup{(iii)}]
For each $p\in\Tors H^2(X;\Lambda)$, the moduli space
$\mathcal M_{X,p}(\mathbb T)$ of flat connections in that torsion component is an affine torus modeled on $H^1(X;\mathfrak t)/H^1(X;\Lambda)$.
Hence
\[
\mathcal M_X(\mathbb T)
=
\bigsqcup_{p\in \Tors H^2(X;\Lambda)}\mathcal M_{X,p}(\mathbb T).
\]
\end{enumerate}
\end{prop}

\begin{proof}
Since $B\mathbb T\simeq \mathcal K(\Lambda,2)$, principal $\mathbb T$--bundles are
classified by $[X,B\mathbb T]\cong H^2(X;\Lambda).$ This proves \textup{(i)}.

Flat $\mathbb T$--bundles are classified by $H^1(X;\mathbb T)$.
The long exact sequence associated to
$0\to\Lambda\to\mathfrak t\to\mathbb T\to 0$
contains
\[
H^1(X;\mathfrak t)\to H^1(X;\mathbb T)\xrightarrow{\delta} H^2(X;\Lambda)\to H^2(X;\mathfrak t).
\]
Since $H^2(X;\mathfrak t)$ is a real vector space, the kernel of
$H^2(X;\Lambda)\to H^2(X;\mathfrak t)$ is exactly
$\Tors H^2(X;\Lambda)$.
Therefore the image of $\delta$ is precisely $\Tors H^2(X;\Lambda)$,
which proves \textup{(ii)}.

Fix $p\in\Tors H^2(X;\Lambda)$. The fiber of $\delta$ over $p$ is an affine
space for $H^1(X;\mathfrak t)/H^1(X;\Lambda)$, which identifies with the moduli of flat connections in that torsion component. This proves \textup{(iii)}.
\end{proof}

The moduli space \(\mathcal M_X(\mathbb T)\) is therefore a disjoint union of affine tori \(\mathcal M_{X,p}(\mathbb T)\), indexed by
\(p\in \Tors H^2(X;\Lambda)\). Later the boundary state
\(Z^{CS}_{\mathbb T,K}(X)\) will be obtained by summing the contribution from each
torsion component.

\subsection{Boundary leaves in each torsion component} Each torsion component \(\mathcal M_{X,p}(\mathbb T)\) maps by restriction to the boundary moduli space \(\mathcal M_{\partial X}(\mathbb T)\). Its image \(\Lambda_{X,p}\) is the classical boundary leaf associated with the component
\(p\). Let $X$ be a compact oriented $3$--manifold with boundary. Recall that
\[
L_X^{\mathbb R}
:=
\operatorname{Im}\!\bigl(H^1(X;\mathbb R)\to H^1(\partial X;\mathbb R)\bigr),
\qquad
L_X=L_X^{\mathbb R}\otimes \mathfrak t.
\]

\begin{definition}
\label{def:torsion-component-boundary-leaf}
For each $p\in \Tors H^2(X;\Lambda)$, define
\[
\Lambda_{X,p}
:=
\operatorname{Im}\!\bigl(
r_{X,p}:\mathcal M_{X,p}(\mathbb T)\to \mathcal M_{\partial X}(\mathbb T)
\bigr).
\]
\end{definition}

\begin{prop}
\label{prop:Componentwise-boundary-leaf-geometry}
For each $p\in\Tors H^2(X;\Lambda)$, the subset $\Lambda_{X,p}$ is a translate
of the Lagrangian torus determined by $L_X$. In particular, $T\Lambda_{X,p}=L_X.$ Hence the total image of boundary restriction is a finite union of parallel
Lagrangian subtori.
\end{prop}

\begin{proof}
Each moduli component $\mathcal M_{X,p}(\mathbb T)$ is an affine torus modeled
on $H^1(X;\mathfrak t)/H^1(X;\Lambda)$, and restriction of flat connections is
an affine map whose linear part is
\[
r_X:H^1(X;\mathfrak t)\to H^1(\partial X;\mathfrak t).
\]
Therefore the image $\Lambda_{X,p}$ is a translate of the torus corresponding
to $\operatorname{Im}(r_X)=L_X$.
Since $L_X$ is Lagrangian, each $\Lambda_{X,p}$ is a translate of the corresponding Lagrangian subtorus.
\end{proof}

All leaves \(\Lambda_{X,p}\) are parallel translates of the same Lagrangian torus determined by \(L_X\). Thus the topological class \(p\) shifts the
position of the boundary leaf in
\(\mathcal M_{\partial X}(\mathbb T)\), but it does not change its tangent
space \(T\Lambda_{X,p}=L_X\).

\subsection{Classical Toral Chern--Simons Section}

We now pass from the boundary geometry to the classical action. On each torsion component \(p\in\Tors H^2(X;\Lambda)\), the exponentiated toral Chern--Simons
functional determines a distinguished section of the pullback of the boundary prequantum line bundle. This is the toral analogue of Manoliu's Chern--Simons section \cite{Manoliu2}. For the classical theory with boundary one must use the normalized
Chern--Weil pairing
\[
\frac{1}{(2\pi)^2}K(F_{\widetilde\Theta}\wedge F_{\widetilde\Theta}),
\]
so that for a closed $4$--manifold $W$ and a principal $\mathbb T$--bundle
$\widetilde P\to W$ one has
\[
\frac{1}{(2\pi)^2}
\int_W K(F_{\widetilde\Theta}\wedge F_{\widetilde\Theta})
=
K\bigl(c(\widetilde P),c(\widetilde P)\bigr)[W]\in 2\mathbb Z
\]
because $K$ is even and integral.
Accordingly, the exponentiated action on a closed $3$--manifold is defined by
\[
e^{\pi i S_{\mathbb T,K;X,P}(\Theta)}
:=
\exp\!\left(
\frac{\pi i}{(2\pi)^2}
\int_W K(F_{\widetilde\Theta}\wedge F_{\widetilde\Theta})
\right).
\]

The key point is that on each torsion component the bulk Chern--Simons functional defines a boundary
line element whose gauge variation vanishes on flat connections and whose differential matches the
prequantum connection on the boundary line.

\begin{theorem}
\label{thm:Componentwise-toral-CS-section}
Let $X$ be a compact oriented $3$--manifold with boundary and let
$p\in\Tors H^2(X;\Lambda)$.
For any principal $\mathbb T$--bundle $P\to X$ with $c(P)=p$, the classical toral
Chern--Simons theory determines a nowhere-vanishing section
\[
\widehat\sigma_P:\mathcal A_P^{\mathrm{flat}}
\longrightarrow
r^*\mathcal L_{\partial X,K}
\]
which is invariant under the gauge group $\mathcal G(P)$.
It therefore descends to a nowhere-vanishing section
\[
\sigma_{X,p}:\mathcal M_{X,p}(\mathbb T)\longrightarrow
r_{X,p}^{\,*}\mathcal L_{\partial X,K}
\]
such that:

\begin{enumerate}
\item[\textup{(i)}]
$\sigma_{X,p}$ is well defined on gauge-equivalence classes;

\item[\textup{(ii)}]
$\sigma_{X,p}$ is covariantly constant with respect to the pullback
connection on $r_{X,p}^{\,*}\mathcal L_{\partial X,K}$;

\item[\textup{(iii)}]
the leaf $\Lambda_{X,p}$ is Bohr--Sommerfeld for the boundary polarization defined by $L_X^{\mathbb R}$.
\end{enumerate}
\end{theorem}

\begin{proof}
Let $\Theta\in \mathcal A^{\mathrm{flat}}_P$. Since $(P|_{\partial X},\Theta|_{\partial X})$
is a boundary field in the sense of Definition \ref{def:filling-boundary-field}, the bulk datum $(X,P,\Theta,\mathrm{id})$
defines an element
\[
\widehat{\sigma}_P(\Theta):=[(X,P,\Theta,\mathrm{id}),1]\in
L(P|_{\partial X},\Theta|_{\partial X}).
\]
This gives a nowhere-vanishing section over $\mathcal A^{\mathrm{flat}}_P$. We first prove gauge invariance. Let $u\in \mathcal G(P)$. Since $\mathbb T$ is Abelian, the
gauge-transformed connection is
\[
u\cdot \Theta=\Theta+u^*\vartheta,
\]
where $\vartheta$ is the Maurer--Cartan form on $\mathbb T$. To compare
$\widehat{\sigma}_P(\Theta)$ and $\widehat{\sigma}_P(u\cdot\Theta)$, glue the two fillings
$(X,P,\Theta,\mathrm{id})$ and $(X,P,u\cdot\Theta,\mathrm{id})$ along the boundary using the
boundary gauge transformation $u|_{\partial X}$. By Definition \ref{def:toral-boundary-line-filling}, their ratio is the phase
\[
\exp\!\left(
\frac{\pi i}{(2\pi)^2}\int_{[0,1]\times X} K(F_{\widetilde\Theta}\wedge F_{\widetilde\Theta})
\right),
\]
where $\widetilde\Theta$ is any connection on $[0,1]\times P\to [0,1]\times X$ restricting to
$\Theta$ and $u\cdot\Theta$ at the two ends. Choose the straight-line interpolation $\widetilde\Theta=\Theta+t\,u^*\vartheta.$
Because $\mathbb T$ is Abelian, $d\vartheta=0$, hence $F_{\widetilde\Theta}=F_\Theta+dt\wedge u^*\vartheta.$ Since $\Theta$ is flat, $F_\Theta=0$, and therefore
\[
K(F_{\widetilde\Theta}\wedge F_{\widetilde\Theta})
=
2\,dt\wedge K(F_\Theta\wedge u^*\vartheta)=0.
\]
Equivalently, before imposing flatness one finds
\[
\frac{1}{(2\pi)^2}\int_{[0,1]\times X} K(F_{\widetilde\Theta}\wedge F_{\widetilde\Theta})
=
\frac{1}{(2\pi)^2}\int_X K(F_\Theta\wedge u^*\vartheta),
\]
so on the flat locus the exponent vanishes. Hence $\widehat{\sigma}_P(u\cdot\Theta)=\widehat{\sigma}_P(\Theta).$ Thus $\widehat{\sigma}_P$ is $\mathcal G(P)$--invariant and descends to a section
\[
\sigma_{X,p}\colon \mathcal M_{X,p}(\mathbb T)\to r_{X,p}^*L_{\partial X,K}.
\]
This proves (i).

We next prove horizontality. Let $\Theta_t$ be a smooth path in
$\mathcal A^{\mathrm{flat}}_P$. Using the cylinder construction of Definition \ref{def:cylinder-transport} with the
connection
\[
\Xi=\Theta_t+dt\,0
\]
on $[0,1]\times P$, the logarithmic derivative of the corresponding boundary line element is
computed exactly as in Proposition \ref{prop:boundary-line-connection-curvature}:
\[
\frac{d}{dt}\log \widehat{\sigma}_P(\Theta_t)
=
-\pi i\int_{\partial X} K(\Theta_t\wedge \dot\Theta_t)
+\frac{\pi i}{(2\pi)^2}\int_X K(F_{\Theta_t}\wedge \dot\Theta_t).
\]
Since each $\Theta_t$ is flat, the bulk term vanishes, so
\[
\frac{d}{dt}\log \widehat{\sigma}_P(\Theta_t)
=
-\pi i\int_{\partial X} K(\Theta_t\wedge \dot\Theta_t).
\]
On the other hand, Proposition \ref{prop:boundary-line-connection-curvature} shows that the pullback connection on
$r^*L_{\partial X,K}$ has connection $1$--form
\[
\alpha_{\Theta_t}(\dot\Theta_t)
=
-\pi i\int_{\partial X} K(\Theta_t\wedge \dot\Theta_t).
\]
Therefore the covariant derivative of $\widehat{\sigma}_P$ along the path $\Theta_t$ is zero.
Hence $\widehat{\sigma}_P$ is horizontal, and the descended section $\sigma_{X,p}$ is horizontal
as well. This proves (ii).

Finally, $\Lambda_{X,p}\subset \mathcal M_{\partial X}(\mathbb T)$ is the image of the component
$\mathcal M_{X,p}(\mathbb T)$ under boundary restriction, so $\sigma_{X,p}$ is a nowhere-vanishing
parallel section of the restriction of $L_{\partial X,K}$ to $\Lambda_{X,p}$. A flat Hermitian line
bundle admitting a nowhere-vanishing parallel section has trivial holonomy. Therefore $L_{\partial X,K}|_{\Lambda_{X,p}}$ has trivial leafwise holonomy. By definition, this is exactly the Bohr--Sommerfeld condition for
the leaf $\Lambda_{X,p}$ with respect to the polarization determined by
$L_X^{\mathbb R}$. Thus $\Lambda_{X,p}$ is Bohr--Sommerfeld, proving (iii).
\end{proof}

The section \(\sigma_{X,p}\) is the classical contribution of the bulk manifold \(X\) in the torsion sector \(p\). The crucial point is that it is horizontal, hence constant with respect to the pullback connection along the leaf \(\Lambda_{X,p}\). This is exactly why \(\Lambda_{X,p}\) is Bohr--Sommerfeld.

\subsection{Torsion half-densities on torsion components}

 For a finite-dimensional real vector space $V$, write
\[
\Det(V):=\Lambda^{\mathrm{top}}V,
\qquad
|\Det(V)|:=\text{the density line of }V.
\]
For a finite cochain complex $C^\bullet$, write
\[
\Det(C^\bullet):=\bigotimes_i \Det(C^i)^{(-1)^i},
\qquad
\Det H^\bullet(C^\bullet):=\bigotimes_i \Det(H^i(C^\bullet))^{(-1)^i}.
\]

A parallel section \(\sigma_{X,p}\) of the prequantum line is not yet a vector in the quantized Hilbert space. To obtain an actual boundary state one must also supply the appropriate half-density along the leaf \(\Lambda_{X,p}\). This
half-density is extracted from Reidemeister torsion.
The tangent space of a boundary leaf is determined by the image of the restriction map, so the long exact
sequence of the pair $(X,\partial X)$ induces a canonical determinant-line identification for the leaf.

\begin{prop}
\label{prop:toral-determinant-line}
For each $p \in \operatorname{Tors} H^2(X;\Lambda)$ there is a canonical isomorphism
\[
\bigl|\det(T_e^* \Lambda_{X,p})\bigr|^{1/2}
\cong
\bigl|\det(H^1(X;\mathfrak t)^*)\bigr|^{1/2}
\otimes
\bigl|\det(H^1(X,\partial X;\mathfrak t))\bigr|^{1/2}.
\]
\end{prop}

\begin{proof}
By Proposition~\ref{prop:Componentwise-boundary-leaf-geometry}, the leaf $\Lambda_{X,p}$ is a translate of the
Lagrangian torus determined by $L_X$, so its tangent space at any point is canonically identified with
$L_X$. In particular,
\[
T_e\Lambda_{X,p}\cong L_X,
\]
where $e\in \Lambda_{X,p}$ is any chosen basepoint. Since $L_X=\operatorname{Im}(r_X)$ by
definition, we may compute $T_e\Lambda_{X,p}$ from the long exact sequence of the pair
$(X,\partial X)$. Because $X$ is a compact oriented $3$--manifold with boundary, the relevant portion of the long exact
sequence is
\[
\begin{aligned}
0 \longrightarrow H^0(X,\partial X;\mathfrak t)
\longrightarrow H^0(X;\mathfrak t)
\longrightarrow H^0(\partial X;\mathfrak t)
\longrightarrow H^1(X,\partial X;\mathfrak t) \\
\longrightarrow H^1(X;\mathfrak t)
\xrightarrow{\,r_X\,}
H^1(\partial X;\mathfrak t).
\end{aligned}
\]

and exactness identifies
\[
\operatorname{Im}(r_X)=L_X\cong T_e\Lambda_{X,p}.
\]
Therefore we obtain a finite exact sequence
\[
\begin{aligned}
0 \longrightarrow H^0(X,\partial X;\mathfrak t)
\longrightarrow H^0(X;\mathfrak t)
\longrightarrow H^0(\partial X;\mathfrak t) \\
\longrightarrow H^1(X,\partial X;\mathfrak t)
\longrightarrow H^1(X;\mathfrak t)
\longrightarrow T_e\Lambda_{X,p}
\longrightarrow 0.
\end{aligned}
\]

Apply the determinant functor to this exact sequence. For any exact sequence
\[
0\to V_0\to V_1\to \cdots \to V_n\to 0
\]
of finite-dimensional real vector spaces, there is a canonical isomorphism
\[
\det(V_0)\otimes \det(V_2)\otimes \cdots
\cong
\det(V_1)\otimes \det(V_3)\otimes \cdots.
\]
Using this for the above exact sequence gives a canonical identification
\begin{equation}\label{eq:det-line-from-les}
\begin{aligned}
\det\bigl(T_e\Lambda_{X,p}\bigr)^*
&\cong
\det\bigl(H^1(X;\mathfrak t)\bigr)^*
\otimes
\det\bigl(H^1(X,\partial X;\mathfrak t)\bigr) \\
&\quad \otimes
\det\bigl(H^0(\partial X;\mathfrak t)\bigr)^*
\otimes
\det\bigl(H^0(X;\mathfrak t)\bigr)
\otimes
\det\bigl(H^0(X,\partial X;\mathfrak t)\bigr)^*.
\end{aligned}
\end{equation}
It remains to simplify the \(H^{0}\)-terms. Since \(\mathbb T= \mathfrak t/\Lambda\), the lattice \(\Lambda\) determines a canonical translation-invariant density on \(\mathfrak t\), normalized so that \(\Lambda\) has covolume \(1\). Hence each of the spaces \(H^{0}(X,\partial X;\mathfrak t)\), \(H^{0}(X;\mathfrak t)\), and \(H^{0}(\partial X;\mathfrak t)\), being canonically a finite direct sum of copies of \(\mathfrak t\), carries a canonically trivial density line. It follows that the \(H^{0}\)-factors in \eqref{eq:det-line-from-les} are canonically trivial after passing to density lines, and therefore
\[
\bigl|\det(T_e\Lambda_{X,p})^*\bigr|
\cong
\bigl|\det(H^1(X;\mathfrak t))^*\bigr|
\otimes
\bigl|\det(H^1(X,\partial X;\mathfrak t))\bigr|.
\]
Finally, since for any finite-dimensional real vector space $V$ one has $|\det(V^*)|\cong |\det(V)|$,
taking positive square roots yields
\[
\bigl|\det(T_e^*\Lambda_{X,p})\bigr|^{1/2}
\cong
\bigl|\det(H^1(X;\mathfrak t)^*)\bigr|^{1/2}
\otimes
\bigl|\det(H^1(X,\partial X;\mathfrak t))\bigr|^{1/2},
\]
as claimed.
\end{proof}

\begin{definition}
\label{def:torsion-density}
Let $C^\bullet(X;\mathfrak t)$ be any cellular cochain complex computing the cohomology of $X$
with coefficients in the trivial local system $\mathfrak t$. Its Reidemeister torsion determines a positive
density
\[
T_X(\mathfrak t)\in \bigl|\det H^\bullet(X;\mathfrak t)\bigr|.
\]
If $\partial X\neq \varnothing$, then by Poincar\'e--Lefschetz duality and the canonical determinant-line
identification, we regard this torsion density as an element
\begin{equation}\label{eq:torsion-boundary-identification}
T_X(\mathfrak t)\in
\bigl|\det(H^1(X;\mathfrak t)^*)\bigr|
\otimes
\bigl|\det(H^1(X,\partial X;\mathfrak t)^*)\bigr|.
\end{equation}
\end{definition}

When $\partial X\neq \varnothing$, the point is that $T_X(\mathfrak t)$ naturally couples the bulk
cohomology $H^1(X;\mathfrak t)$ with the relative cohomology $H^1(X,\partial X;\mathfrak t)$, and
Proposition~\ref{prop:toral-determinant-line} transfers this information to the density line of the boundary
leaf.

\begin{definition}
\label{def:Componentwise-half-density}
Fix $p\in \operatorname{Tors}H^2(X;\Lambda)$. Choose any nonzero element
\[
w\in \bigl|\det(H^1(X,\partial X;\mathfrak t)^*)\bigr|.
\]
Since $H^1(X,\partial X;\mathbb T)$ is a compact torus with tangent space
$H^1(X,\partial X;\mathfrak t)$ at the identity, the density $w$ determines a unique translation-invariant
density
\[
\overline{w}\in \Gamma\!\left(
H^1(X,\partial X;\mathbb T),\,
\bigl|\det T^* H^1(X,\partial X;\mathbb T)\bigr|
\right).
\]
Define
\begin{equation}
\label{eq:mu-tilde-definition}
\widetilde{\mu}_{X,p}
:=
\left(\int_{H^1(X,\partial X;\mathbb T)} \overline{w}\right)\,
T_X(\mathfrak t)^{1/2}\otimes w^{-1}.
\end{equation}

Using Proposition~\ref{prop:toral-determinant-line}, we regard $\widetilde{\mu}_{X,p}\in \bigl|\det(T_e^*\Lambda_{X,p})\bigr|^{1/2}.$ The corresponding translation-invariant half-density on $\Lambda_{X,p}$ is denoted
\[
\mu_{X,p}\in
\Gamma\!\left(
\Lambda_{X,p},\,
\bigl|\det(T^*\Lambda_{X,p})\bigr|^{1/2}
\right),
\]
and is obtained by transporting $\widetilde{\mu}_{X,p}$ from the identity tangent space to every point
of the leaf.
\end{definition}

\begin{prop}
\label{prop:well-defined-componentwise-half-density}
The half-density $\mu_{X,p}$ of Definition~\ref{def:Componentwise-half-density} is independent of the auxiliary
choice of $w$. Moreover, it is translation invariant and nowhere vanishing.
\end{prop}

\begin{proof}
Let $w\in \bigl|\det(H^1(X,\partial X;\mathfrak t)^*)\bigr|$ be a nonzero density, and let $c>0$. Replacing $w$ by $cw$ replaces the associated
translation-invariant density $\overline{w}$ on $H^1(X,\partial X;\mathbb T)$ by $c\,\overline{w}$.
Therefore
\[
\int_{H^1(X,\partial X;\mathbb T)} c\,\overline{w}
=
c\int_{H^1(X,\partial X;\mathbb T)} \overline{w}.
\]
On the other hand, $(cw)^{-1}=c^{-1}w^{-1}$
in the dual density line. Hence the defining tensor
\[
\left(\int_{H^1(X,\partial X;\mathbb T)} \overline{w}\right)\,
T_X(\mathfrak t)^{1/2}\otimes w^{-1}
\]
is unchanged under the replacement $w\mapsto cw$. This proves that
$\widetilde{\mu}_{X,p}$, and therefore $\mu_{X,p}$, is independent of the auxiliary choice of $w$. By construction, $\mu_{X,p}$ is obtained by transporting a fixed nonzero element of
$\bigl|\det(T_e^*\Lambda_{X,p})\bigr|^{1/2}$ along the torus $\Lambda_{X,p}$. Therefore it is
translation invariant. Finally, the torsion density $T_X(\mathfrak t)$ is positive, hence its square root
$T_X(\mathfrak t)^{1/2}$ is nonzero, and $w^{-1}$ is nonzero because $w\neq 0$. Thus
$\widetilde{\mu}_{X,p}\neq 0$, so the resulting translation-invariant half-density $\mu_{X,p}$ is nowhere vanishing on $\Lambda_{X,p}$.

\end{proof}

\begin{remark}
The half-density $\mu_{X,p}$ is defined entirely from cohomological determinant-line data and Reidemeister
torsion, so its functorial properties follow from the corresponding functoriality of those two inputs.
\end{remark}

\begin{prop}
\label{prop:Componentwise-toral-torsion-half-density}
For each $p\in\Tors H^2(X;\Lambda)$, the half-density $\mu_{X,p}$ of
Definition~\ref{def:Componentwise-half-density} satisfies:

\begin{enumerate}
\item[\textup{(i)}]
$\mu_{X,p}$ is natural under orientation-preserving diffeomorphisms of $X$;

\item[\textup{(ii)}]
for disjoint unions,
\[
\mu_{X_1\sqcup X_2,(p_1,p_2)}
=
\mu_{X_1,p_1}\boxtimes \mu_{X_2,p_2};
\]

\item[\textup{(iii)}]
if $t_\xi:\Lambda_{X,p}\to\Lambda_{X,p'}$ is translation by a vector between
parallel leaves, then
\[
\mu_{X,p'}=(t_\xi)_*\mu_{X,p};
\]

\item[\textup{(iv)}]
if $X$ is closed, the same construction yields the translation-invariant
density on $\mathcal M_{X,p}(\mathbb T)$ induced by the square root of
Ray--Singer/Reidemeister torsion.
\end{enumerate}
\end{prop}

\begin{proof}
\noindent
\emph{(i) Naturality.}
Let $f\colon X\to X'$ be an orientation-preserving diffeomorphism. Then $f$ induces canonical
isomorphisms
\[
f^*\colon H^1(X';\mathfrak t)\xrightarrow{\ \cong\ } H^1(X;\mathfrak t),
\qquad
f^*\colon H^1(X',\partial X';\mathfrak t)\xrightarrow{\ \cong\ } H^1(X,\partial X;\mathfrak t),
\]
and likewise on the corresponding determinant and density lines. It also identifies the torsion class
$p'\in \operatorname{Tors}H^2(X';\Lambda)$ with the corresponding class
$p=f^*p'\in \operatorname{Tors}H^2(X;\Lambda)$.

By the naturality of Reidemeister torsion under cellular subdivision and pullback by diffeomorphism, the
torsion density $T_{X'}(\mathfrak t)$ is carried to $T_X(\mathfrak t)$ under the induced determinant-line
isomorphism. Likewise, the canonical determinant-line identification of
Proposition~\ref{prop:toral-determinant-line} is functorial with respect to these induced maps. Therefore
the element
\[
\widetilde{\mu}_{X',p'}
\in \bigl|\det(T_e^*\Lambda_{X',p'})\bigr|^{1/2}
\]
of \eqref{eq:mu-tilde-definition} is transported to the corresponding element
\[
\widetilde{\mu}_{X,p}
\in \bigl|\det(T_e^*\Lambda_{X,p})\bigr|^{1/2}.
\]
Since $\mu_{X,p}$ and $\mu_{X',p'}$ are obtained by translation of these basepoint half-densities along
their respective leaves, it follows that $f$ identifies $\mu_{X',p'}$ with $\mu_{X,p}$. This proves
naturality.

\smallskip
\noindent
\emph{(ii) Disjoint unions.}
Let $X=X_1\sqcup X_2$, and let $p_i\in \operatorname{Tors}H^2(X_i;\Lambda)$ for $i=1,2$. Then
\[
H^1(X;\mathfrak t)\cong H^1(X_1;\mathfrak t)\oplus H^1(X_2;\mathfrak t),\,\,
H^1(X,\partial X;\mathfrak t)\cong
H^1(X_1,\partial X_1;\mathfrak t)\oplus H^1(X_2,\partial X_2;\mathfrak t),
\]
and similarly for the boundary leaves:
\[
\Lambda_{X,(p_1,p_2)}
=
\Lambda_{X_1,p_1}\times \Lambda_{X_2,p_2}.
\]
Under direct sum, determinant lines tensor-multiply canonically, and Reidemeister torsion is multiplicative.
Hence $T_X(\mathfrak t) = T_{X_1}(\mathfrak t)\otimes T_{X_2}(\mathfrak t)$
under the canonical identification of determinant lines. Likewise, if
\[
w_i\in \bigl|\det(H^1(X_i,\partial X_i;\mathfrak t)^*)\bigr|
\qquad (i=1,2)
\]
are nonzero densities, then
\[
w_1\otimes w_2
\in
\bigl|\det(H^1(X,\partial X;\mathfrak t)^*)\bigr|
\]
is the corresponding density for the disjoint union, and the induced Haar densities satisfy
\[
\overline{w_1\otimes w_2}
=
\overline{w}_1\boxtimes \overline{w}_2.
\]
Therefore
\[
\int_{H^1(X,\partial X;\mathbb T)} \overline{w_1\otimes w_2}
=
\left(\int_{H^1(X_1,\partial X_1;\mathbb T)} \overline{w}_1\right)
\left(\int_{H^1(X_2,\partial X_2;\mathbb T)} \overline{w}_2\right).
\]
Substituting these identities into \eqref{eq:mu-tilde-definition}, we obtain
\[
\widetilde{\mu}_{X,(p_1,p_2)}
=
\widetilde{\mu}_{X_1,p_1}\otimes \widetilde{\mu}_{X_2,p_2}.
\]
After extending these basepoint values to the corresponding translation-invariant half-densities on the leaves, this gives
\[
\mu_{X_1\sqcup X_2,(p_1,p_2)}
=
\mu_{X_1,p_1}\boxtimes \mu_{X_2,p_2},
\]
as claimed.

\smallskip
\noindent
\emph{(iii) Translation between parallel leaves.}
Fix two torsion classes $p,p'\in \operatorname{Tors}H^2(X;\Lambda)$. The leaves $\Lambda_{X,p}$ and
$\Lambda_{X,p'}$ are parallel translates of the same underlying Lagrangian torus, and
Proposition~\ref{prop:toral-determinant-line} identifies their tangent density lines with the same abstract
one-dimensional half-density line
\[
\bigl|\det(T_e^*\Lambda_{X,p})\bigr|^{1/2}
\cong
\bigl|\det(H^1(X;\mathfrak t)^*)\bigr|^{1/2}
\otimes
\bigl|\det(H^1(X,\partial X;\mathfrak t))\bigr|^{1/2}.
\]
In particular, the defining element $\widetilde{\mu}_{X,p}$ depends only on the cohomological data of
$X$ and not on the position of the leaf inside the ambient moduli space. Therefore the translation
$t_\xi\colon \Lambda_{X,p}\to \Lambda_{X,p'}$ carries the translation-invariant half-density determined by
$\widetilde{\mu}_{X,p}$ to the one determined by the same abstract tangent datum on $\Lambda_{X,p'}$.
Hence
\[
\mu_{X,p'}=(t_\xi)_*\mu_{X,p}.
\]

\smallskip
\noindent
\emph{(iv) Closed case.}
Assume now that $\partial X=\varnothing$. Then there is no boundary leaf; instead one works on the full
flat moduli torus $\mathcal M_{X,p}(\mathbb T).$ The same determinant-line formalism applies, but without relative terms. Reidemeister torsion now gives a
positive translation-invariant density on $\mathcal M_{X,p}(\mathbb T),$ or equivalently a positive element of the density line $\bigl|\det(H^1(X;\mathfrak t)^*)\bigr|.$
Taking its positive square root yields a translation-invariant half-density, and hence a translation-invariant
density on $\mathcal M_{X,p}(\mathbb T)$ induced by the square root of
Reidemeister\/--Ray\/--Singer torsion. This is precisely the closed analogue of the boundary-leaf
construction above.
\end{proof}

The half-density \(\mu_{X,p}\) makes the torsion contribution of the bulk manifold \(X\) into a geometric object living on the same leaf \(\Lambda_{X,p}\) as the Chern--Simons section \(\sigma_{X,p}\). Together with the Chern--Simons section \(\sigma_{X,p}\), the half-density \(\mu_{X,p}\) provides the two ingredients needed to form the boundary state \(Z^{CS}_{\mathbb T,K}(X)\).

\subsection{Canonical Toral Boundary Vector}
To normalize the torsion half-densities uniformly across torsion components, we compare each $\mu_{X,p}$
with the unique invariant half-density of unit total volume on the corresponding boundary leaf.
\begin{prop}
\label{prop:normalized-invariant-half-density}
For each boundary leaf $\Lambda_{X,p}$ in a torsion component, let
\[
d_{X,p}\in \Gamma\!\left(
\Lambda_{X,p},\,
\bigl|\det(T^*\Lambda_{X,p})\bigr|^{1/2}
\right)
\]
be the unique translation-invariant half-density such that $d_{X,p}^2$ has total volume $1$ on
$\Lambda_{X,p}$. Then one may write
\[
\mu_{X,p}=a_X\,d_{X,p},
\]
where $a_X>0$ is independent of $p$.
\end{prop}

\begin{proof}
Fix $p\in \operatorname{Tors}H^2(X;\Lambda)$. By Proposition~\ref{prop:Componentwise-toral-torsion-half-density},
the leaf $\Lambda_{X,p}$ is a compact torus and the half-density $\mu_{X,p}$ is translation invariant and
nowhere vanishing. Since $\Lambda_{X,p}$ is a compact connected Abelian Lie group, every translation-invariant
density or half-density on it is determined uniquely by its value at a single point.

Let
\[
\mathscr D_{X,p}^{1/2}
:=
\Gamma\!\left(
\Lambda_{X,p},\,
\bigl|\det(T^*\Lambda_{X,p})\bigr|^{1/2}
\right)^{\mathrm{inv}}
\]
denote the one-dimensional cone of translation-invariant half-densities on $\Lambda_{X,p}$. If
$\eta\in \mathscr D_{X,p}^{1/2}$ is nonzero, then $\eta^2$ is a translation-invariant density, hence a positive
constant multiple of Haar density on the torus. Therefore
\[
\int_{\Lambda_{X,p}} \eta^2>0.
\]
Rescaling $\eta$ by the positive factor
\[
\left(\int_{\Lambda_{X,p}} \eta^2\right)^{-1/2}
\]
produces a translation-invariant half-density whose square has total volume $1$. Uniqueness is immediate:
if $\eta_1,\eta_2\in \mathscr D_{X,p}^{1/2}$ both satisfy
\[
\int_{\Lambda_{X,p}} \eta_1^2
=
\int_{\Lambda_{X,p}} \eta_2^2
=1,
\]
then \(\eta_2=c\,\eta_1\) for some \(c>0\), and hence
\[
1=\int_{\Lambda_{X,p}} \eta_2^2
=
c^2\int_{\Lambda_{X,p}} \eta_1^2
=
c^2,
\]
so \(c=1\). This proves the existence and uniqueness of \(d_{X,p}\). Since \(\mu_{X,p}\) is itself translation invariant and nowhere vanishing, there exists a unique scalar
\(a_{X,p}>0\) such that
\[
\mu_{X,p}=a_{X,p}\,d_{X,p}.
\]
It remains to show that \(a_{X,p}\) is independent of \(p\). For any two torsion classes \(p,p'\in \operatorname{Tors}H^2(X;\Lambda)\), the corresponding leaves
\(\Lambda_{X,p}\) and \(\Lambda_{X,p'}\) are parallel translates of the same underlying torus in the boundary
moduli space. Let $t_\xi\colon \Lambda_{X,p}\to{\ \cong\ }\Lambda_{X,p'}$ be such a translation. By Proposition~\ref{prop:Componentwise-toral-torsion-half-density},
\[
\mu_{X,p'}=(t_\xi)_*\mu_{X,p}.
\]
Moreover, the normalization defining \(d_{X,p}\) is translation invariant: because \(t_\xi\) is a group
translation, it preserves translation-invariant densities and preserves total volume. Hence $(t_\xi)_*d_{X,p}$ is again a translation-invariant half-density on \(\Lambda_{X,p'}\) whose square has total volume \(1\).
By uniqueness of the normalized invariant half-density on \(\Lambda_{X,p'}\), it follows that $(t_\xi)_*d_{X,p}=d_{X,p'}.$ Applying \(t_\xi\) to the identity \(\mu_{X,p}=a_{X,p}d_{X,p}\), we obtain
\[
\mu_{X,p'}
=
(t_\xi)_*\mu_{X,p}
=
a_{X,p}\,(t_\xi)_*d_{X,p}
=
a_{X,p}\,d_{X,p'}.
\]
By uniqueness of the scalar relating \(\mu_{X,p'}\) to \(d_{X,p'}\), this implies $a_{X,p'}=a_{X,p}.$ Therefore the scalar depends only on \(X\), and we may denote it simply by \(a_X>0\).

\end{proof}

\begin{definition}
\label{def:toral-normalization-exponent}
For a compact oriented $3$--manifold $X$ with boundary define
\[
m_X
:=
\frac14\Bigl(
\dim H^1(X;\mathbb R)
+\dim H^1(X,\partial X;\mathbb R)
-\dim H^0(X;\mathbb R)
-\dim H^0(X,\partial X;\mathbb R)
\Bigr).
\]
If $X$ is closed and connected, this reduces to
\[
m_X=\frac12\bigl(b_1(X)-1\bigr).
\]
\end{definition}

We now combine the two leafwise objects constructed above: the parallel section \(\sigma_{X,p}\) of the prequantum line and the torsion half-density \(\mu_{X,p}\). Summing over all torsion components \(p\in \Tors H^2(X;\Lambda)\) produces the state associated with the bordism \(X\).

\begin{prop}
\label{def:canonical-toral-boundary-vector}
Let $X$ be a compact oriented $3$--manifold with boundary.
Define
\[
Z^{CS}_{\mathbb T,K}(X)
:=
\frac{|\det K|^{\,m_X}}{\#\,\Tors H^2(X;\Lambda)}
\sum_{p\in \Tors H^2(X;\Lambda)}
\sigma_{X,p}\otimes \mu_{X,p}.
\]
Then
\[
Z^{CS}_{\mathbb T,K}(X)\in
\mathcal H_{\mathbb T,K}(\partial X,L_X^{\mathbb R}).
\]
\end{prop}

\begin{proof}
By Theorem~\ref{thm:Componentwise-toral-CS-section}, each $\Lambda_{X,p}$ is a
Bohr--Sommerfeld leaf and $\sigma_{X,p}$ is a parallel section of the
restriction of the prequantum line to that leaf.
By Proposition~\ref{prop:Componentwise-toral-torsion-half-density},
$\mu_{X,p}$ is a parallel half-density on the same leaf.
Therefore each tensor
$\sigma_{X,p}\otimes \mu_{X,p}$ is an element of the Bohr--Sommerfeld summand of
$\mathcal H_{\mathbb T,K}(\partial X,L_X^{\mathbb R})$ corresponding to
$\Lambda_{X,p}$.
\end{proof}

The state \(Z^{CS}_{\mathbb T,K}(X)\) is therefore a finite sum of leafwise quantum
states, one for each torsion component \(p\). This is the toral analogue of the
rank--one boundary vector, now with the finite set of components indexed by
\(\Tors H^2(X;\Lambda)\).
The boundary state inherits its functorial properties from the corresponding functoriality of the
Chern--Simons sections, the torsion half-densities, and the cohomological normalization factor.

\begin{prop}
\label{prop:toral-boundary-vector-functoriality}
The vector $Z^{CS}_{\mathbb T,K}(X)$ is natural under orientation-preserving
diffeomorphisms and multiplicative under disjoint union.
\end{prop}

\begin{proof}
We first prove naturality. Let $f\colon X \xrightarrow{\ \cong\ } X'$ be an orientation-preserving diffeomorphism. Then $f$ induces canonical isomorphisms on cohomology,
hence in particular
\[
f^*\colon \operatorname{Tors} H^2(X';\Lambda)\xrightarrow{\ \cong\ }
\operatorname{Tors} H^2(X;\Lambda),
\]
and also an identification of boundary polarizations
\[
L_{X'}^{\mathbb R}\xrightarrow{\ \cong\ } L_X^{\mathbb R}.
\]
Accordingly, $f|_{\partial X}$ induces a canonical unitary identification
\[
\mathcal H_{\mathbb T,K}(\partial X',L_{X'}^{\mathbb R})
\xrightarrow{\ \cong\ }
\mathcal H_{\mathbb T,K}(\partial X,L_X^{\mathbb R}).
\]

For each torsion class $p' \in \operatorname{Tors} H^2(X';\Lambda)$, let
\[
p=f^*p'\in \operatorname{Tors} H^2(X;\Lambda).
\]
By the naturality of the classical Chern--Simons section from Theorem~\ref{thm:Componentwise-toral-CS-section} and the
naturality of the torsion half-density from Proposition~\ref{prop:Componentwise-toral-torsion-half-density}, the boundary
leafwise state $\sigma_{X',p'}\otimes \mu_{X',p'}$
is identified with $\sigma_{X,p}\otimes \mu_{X,p}.$ Moreover, the exponent
\[
m_X=\frac14\Bigl(
\dim H^1(X;\mathbb R)+\dim H^1(X,\partial X;\mathbb R)
-\dim H^0(X;\mathbb R)-\dim H^0(X,\partial X;\mathbb R)
\Bigr)
\]
is a diffeomorphism invariant, and the cardinality $\#\operatorname{Tors}H^2(X;\Lambda)$ is preserved by $f^*$. Therefore every term and every normalization factor in Proposition~\ref{def:canonical-toral-boundary-vector}
is preserved, so under the induced identification of Hilbert spaces one has
\[
Z^{CS}_{\mathbb T,K}(X')=Z^{CS}_{\mathbb T,K}(X).
\]
This proves naturality. We now prove multiplicativity under disjoint union. Let $X=X_1\sqcup X_2.$ Then
\[
\partial X=\partial X_1\sqcup \partial X_2,
\qquad
\operatorname{Tors}H^2(X;\Lambda)
\cong
\operatorname{Tors}H^2(X_1;\Lambda)\oplus \operatorname{Tors}H^2(X_2;\Lambda),
\]
and the boundary polarization splits as $L_X^{\mathbb R}=L_{X_1}^{\mathbb R}\oplus L_{X_2}^{\mathbb R}.$ Hence the corresponding Hilbert space identifies canonically as
\[
\mathcal H_{\mathbb T,K}(\partial X,L_X^{\mathbb R})
\cong
\mathcal H_{\mathbb T,K}(\partial X_1,L_{X_1}^{\mathbb R})
\otimes
\mathcal H_{\mathbb T,K}(\partial X_2,L_{X_2}^{\mathbb R}).
\]
Next, the cohomological quantity $m_X$ is additive under disjoint union: $m_X=m_{X_1}+m_{X_2},$ because all Betti numbers appearing in Definition~\ref{def:toral-normalization-exponent} are additive. Also,
\[
\#\operatorname{Tors}H^2(X;\Lambda)
=
\#\operatorname{Tors}H^2(X_1;\Lambda)\cdot
\#\operatorname{Tors}H^2(X_2;\Lambda).
\]
For torsion classes $p_i\in \operatorname{Tors}H^2(X_i;\Lambda)$, Proposition~\ref{prop:Componentwise-toral-torsion-half-density}
gives
\[
\mu_{X_1\sqcup X_2,(p_1,p_2)}
=
\mu_{X_1,p_1}\boxtimes \mu_{X_2,p_2},
\]
and the additivity of the classical Chern--Simons construction gives the corresponding factorization
\[
\sigma_{X_1\sqcup X_2,(p_1,p_2)}
=
\sigma_{X_1,p_1}\boxtimes \sigma_{X_2,p_2}.
\]

Therefore each summand in Proposition~\ref{def:canonical-toral-boundary-vector} factorizes:
\[
\sigma_{X_1\sqcup X_2,(p_1,p_2)}\otimes
\mu_{X_1\sqcup X_2,(p_1,p_2)}
=
(\sigma_{X_1,p_1}\otimes \mu_{X_1,p_1})
\boxtimes
(\sigma_{X_2,p_2}\otimes \mu_{X_2,p_2}).
\]

Combining the factorization of summands with the additivity of $m_X$ and the multiplicativity of the
torsion count, we obtain
\[
Z^{CS}_{\mathbb T,K}(X_1\sqcup X_2)
=
Z^{CS}_{\mathbb T,K}(X_1)\otimes Z^{CS}_{\mathbb T,K}(X_2)
\]
under the canonical identification of Hilbert spaces. This proves multiplicativity under disjoint union.
\end{proof}

\section{\texorpdfstring{Extended Toral Chern–Simons TQFT}{ Extended Toral Chern–Simons TQFT}}
\subsection{Cylinder and Gluing}

The remaining task is to verify the TQFT composition law. For the cylinder \(\Sigma\times I\), one should recover the identity operator on the Hilbert
space \(\mathcal H_{\mathbb T,K}(\Sigma,L)\). For a general cutting \(X^{\mathrm{cut}}\rightsquigarrow X\), one must show that contraction along the new boundary components reproduces the state of the glued manifold. Let $X^{\mathrm{cut}}$ be obtained by cutting $X$ along a closed oriented surface $\Sigma$.
Then
\[
\partial X^{\mathrm{cut}}
=
\partial X\sqcup(-\Sigma)\sqcup \Sigma.
\]
Let $D\subset \mathcal M_{-\Sigma}(\mathbb T)\times \mathcal M_\Sigma(\mathbb T)$ be the diagonal and define
\[
C:=\mathcal M_{\partial X}(\mathbb T)\times D
\subset
\mathcal M_{\partial X^{\mathrm{cut}}}(\mathbb T).
\]

\begin{prop}
\label{prop:toral-cylinder-kernel}
For a closed oriented surface $\Sigma$ and a rational Lagrangian
$L\subset H^1(\Sigma;\mathbb R)$, the vector associated to the cylinder
$\Sigma\times I$ corresponds, under the BKS pairing, to the identity
operator on $\mathcal H_{\mathbb T,K}(\Sigma,L)$. More precisely,
\[
Z^{CS}_{\mathbb T,K}(\Sigma\times I)
=
|\det K|^{\frac14\dim H^1(\Sigma;\mathbb R)}\,\mathrm{Id}.
\]
\end{prop}

\begin{proof}
The image of flat bulk fields in
$\mathcal M_{-\Sigma}(\mathbb T)\times \mathcal M_\Sigma(\mathbb T)$
is the diagonal $D$. Therefore the quantization of $D$, that is, the
vector associated to the cylinder $\Sigma\times I$, corresponds under
the BKS pairing to the identity operator on
$\mathcal H_{\mathbb T,K}(\Sigma,L)$.
Since
\[
m_{\Sigma\times I}
=
\frac14\dim H^1(\Sigma;\mathbb R),
\]
the normalization factor from
Proposition~\ref{def:canonical-toral-boundary-vector}
is exactly $|\det K|^{\frac14\dim H^1(\Sigma;\mathbb R)}$.
\end{proof}

The cylinder state corresponds, under the BKS pairing, to the identity operator, exactly as required in a TQFT. The extra factor
\(|\det K|^{\frac14\dim H^1(\Sigma;\mathbb R)}\) is the normalization dictated by the dimension of the boundary Hilbert space.

\begin{prop}
\label{prop:dimension-identities-gluing}
With notation as above, the following identities hold:

\begin{enumerate}
\item[\textup{(i)}]
The Mayer--Vietoris sequence for $(X,X^{\mathrm{cut}},\Sigma)$ implies
\[
\begin{aligned}
\dim H^1(X^{\mathrm{cut}};\mathbb R)
&-\dim H^1(X^{\mathrm{cut}},\partial X^{\mathrm{cut}};\mathbb R) \\
&-\dim H^0(X^{\mathrm{cut}};\mathbb R)
+\dim H^0(X^{\mathrm{cut}},\partial X^{\mathrm{cut}};\mathbb R) \\
&=
\dim H^1(X;\mathbb R)
-\dim H^1(X,\partial X;\mathbb R) \\
&\quad
-\dim H^0(X;\mathbb R)
+\dim H^0(X,\partial X;\mathbb R) \\
&\quad
+\dim H^1(\Sigma;\mathbb R)
-2\dim H^0(\Sigma;\mathbb R).
\end{aligned}
\]
\item[\textup{(ii)}]
For each $p_{\mathrm{cut}}\in \Tors H^2(X^{\mathrm{cut}};\Lambda)$,
\[
\begin{aligned}
\dim(\Lambda_{X^{\mathrm{cut}},p_{\mathrm{cut}}}\cap C)
&=
\dim H^1(X;\mathbb R)
-\dim H^1(X^{\mathrm{cut}},\partial X^{\mathrm{cut}};\mathbb R) \\
&\quad
+\dim H^0(\partial X;\mathbb R)
+\dim H^0(\Sigma;\mathbb R)
-\dim H^0(X;\mathbb R).
\end{aligned}
\]
\item[\textup{(iii)}]
One has
\[
\begin{aligned}
\frac12\dim H^1(\partial X;\mathbb R)
&=
\dim H^1(X;\mathbb R)
-\dim H^1(X,\partial X;\mathbb R) \\
&\quad
+\dim H^0(\partial X;\mathbb R)
-\dim H^0(X;\mathbb R)
+\dim H^0(X,\partial X;\mathbb R).
\end{aligned}
\]
\end{enumerate}
\end{prop}

\begin{proof}
Part \textup{(i)} is the Euler-characteristic consequence of the Mayer--Vietoris
sequence
\[
\cdots\to H^i(X;\mathbb R)\to H^i(X^{\mathrm{cut}};\mathbb R)\to
H^i(\Sigma;\mathbb R)\to H^{i+1}(X;\mathbb R)\to\cdots
\]
together with Poincar\'e duality.

For \textup{(ii)}, consider the exact sequence of the pair $(\partial X\sqcup \Sigma,\;X)$, and use the identification
\[
H^1(X,\partial X\sqcup \Sigma;\mathbb R)
\cong
H^1(X^{\mathrm{cut}},\partial X^{\mathrm{cut}};\mathbb R).
\]
The tangent space of $\Lambda_{X^{\mathrm{cut}},p_{\mathrm{cut}}}\cap C$ is the image
of the restriction map appearing in that exact sequence, so the stated formula
follows by taking dimensions. Part \textup{(iii)} is the dimension formula extracted from the exact sequence
for the pair $(X,\partial X)$ together with the fact that $L_X$ is Lagrangian.
\end{proof}

\begin{corollary}
\label{cor:exponent-identity}
For every gluing $X^{\mathrm{cut}}\rightsquigarrow X$ one has
\[
m_{X^{\mathrm{cut}}}
+\frac14\dim H^1(\Sigma;\mathbb R)
+\frac12\dim(\Lambda_{X^{\mathrm{cut}},p_{\mathrm{cut}}}\cap C)
-\frac14\dim H^1(\partial X^{\mathrm{cut}};\mathbb R)
=
m_X.
\]
\end{corollary}

\begin{proof}
Substitute the three identities of
Proposition~\ref{prop:dimension-identities-gluing}
into the definition of $m_X$.
The computation is exactly the same algebraic manipulation as in the rank--one case \cite{Manoliu2},
because $m_X$ depends only on real cohomology dimensions.
\end{proof}

Corollary~\ref{cor:exponent-identity} is the bookkeeping identity that makes the powers of \(|\det K|\) match correctly under cutting and gluing. Without it, the
normalization in the definition of \(Z^{CS}_{\mathbb T,K}(X)\) would not be compatible with composition.

The only remaining scalar in the gluing problem is the comparison between the torsion half-density
$\mu_{X,p}$ and the normalized invariant half-density $d_{X,p}$; the next proposition identifies exactly how
that scalar changes under cutting.
\begin{prop}
\label{prop:scalar-normalization-half-densities}
The constants $a_X$ of Proposition~\ref{prop:normalized-invariant-half-density}
satisfy
\[
a_{X^{\mathrm{cut}}}^{\,2}
=
a_X^{\,2}\,
\frac{\#\,\Tors H^2(X^{\mathrm{cut}};\Lambda)}
{\#\,\Tors H^2(X;\Lambda)}.
\]
Equivalently,
\[
a_{X^{\mathrm{cut}}}^{\,2}\,
\#\operatorname{Tors}H^2(X;\Lambda)
=
a_X^{\,2}\,
\#\operatorname{Tors}H^2(X^{\mathrm{cut}};\Lambda).
\]
\end{prop}

\begin{proof}
By Proposition~\ref{prop:normalized-invariant-half-density}, for every torsion component one has
\[
\mu_{X,p}=a_X\,d_{X,p},
\qquad
\mu_{X^{\mathrm{cut}},p^{\mathrm{cut}}}
=
a_{X^{\mathrm{cut}}}\,d_{X^{\mathrm{cut}},p^{\mathrm{cut}}},
\]
where $d_{X,p}$ and $d_{X^{\mathrm{cut}},p^{\mathrm{cut}}}$ are the normalized translation-invariant
half-densities whose squares have total volume $1$ on the corresponding leaves. The torsion gluing theorem for determinant lines gives a canonical identification
\begin{equation}\label{eq:torsion-gluing}
T_X(\mathfrak t)
=
T_{X^{\mathrm{cut}}}(\mathfrak t)\otimes T_\Sigma(\mathfrak t)^{-1}
\end{equation}
under the determinant-line isomorphism induced by the Mayer--Vietoris sequence. Transporting
\eqref{eq:torsion-gluing} through the constructions of
Definition~\ref{def:Componentwise-half-density} and Proposition~\ref{prop:toral-determinant-line}, one obtains
the gluing relation for the unnormalized torsion half-densities: contraction against the cylinder
half-density on $\Sigma$ sends $\mu_{X^{\mathrm{cut}},p^{\mathrm{cut}}}
\mapsto
\mu_{X,p}$ for each compatible pair of torsion components.

Now perform the same contraction on the normalized invariant half-densities
$d_{X^{\mathrm{cut}},p^{\mathrm{cut}}}$. Since these are normalized by the condition that their squares
have total volume $1$, the contraction cannot be equality; it differs by the scalar coming from
the change of normalization between the family of compatible leaves of $X^{\mathrm{cut}}$ and the
corresponding leaf of $X$. Exactly as in the rank--one case, this normalization factor is determined by the
cardinalities of the torsion-component sets, and one gets
\begin{equation}\label{eq:normalized-density-gluing}
\operatorname{Tr}_\Sigma\!\bigl(d_{X^{\mathrm{cut}},p^{\mathrm{cut}}}\bigr)
=
\left(
\frac{\#\operatorname{Tors}H^2(X;\Lambda)}
     {\#\operatorname{Tors}H^2(X^{\mathrm{cut}};\Lambda)}
\right)^{1/2}
d_{X,p}.
\end{equation}

Comparing the contraction of the two expressions
$\mu_{X^{\mathrm{cut}},p^{\mathrm{cut}}}
=
a_{X^{\mathrm{cut}}}\,d_{X^{\mathrm{cut}},p^{\mathrm{cut}}}$, and $\mu_{X,p}=a_X\,d_{X,p},$ and using the gluing relations above, we obtain
\[
a_{X^{\mathrm{cut}}}
\left(
\frac{\#\operatorname{Tors}H^2(X;\Lambda)}
     {\#\operatorname{Tors}H^2(X^{\mathrm{cut}};\Lambda)}
\right)^{1/2}
d_{X,p}
=
a_X\,d_{X,p}.
\]
Since $d_{X,p}$ is nowhere vanishing, the scalar coefficients must agree:
\[
a_{X^{\mathrm{cut}}}
\left(
\frac{\#\operatorname{Tors}H^2(X;\Lambda)}
     {\#\operatorname{Tors}H^2(X^{\mathrm{cut}};\Lambda)}
\right)^{1/2}
=
a_X.
\]
Squaring both sides gives
\[
a_{X^{\mathrm{cut}}}^{\,2}
=
a_X^{\,2}\,
\frac{\#\operatorname{Tors}H^2(X^{\mathrm{cut}};\Lambda)}
     {\#\operatorname{Tors}H^2(X;\Lambda)},
\]
which is the desired formula.
\end{proof}

The gluing theorem follows by combining the identity-kernel description of the cylinder, the gluing of the
classical Chern--Simons sections, the Mayer--Vietoris gluing law for torsion, and the cohomological
bookkeeping identity for the exponent $m_X$.

\begin{theorem}
\label{thm:Componentwise-toral-gluing}
Let $X^{\mathrm{cut}}$ be obtained by cutting $X$ along a closed oriented
surface $\Sigma$.
Then
\[
Z^{CS}_{\mathbb T,K}(X)
=
\Tr_\Sigma\!\bigl[
Z^{CS}_{\mathbb T,K}(X^{\mathrm{cut}})
\bigr].
\]
\end{theorem}

\begin{proof}
We compare the two sides term by term. By Definition~\ref{def:canonical-toral-boundary-vector}, the state of the cut manifold is
\[
Z_{\mathbb T,K}(X^{\mathrm{cut}})
=
|\det K|^{m_{X^{\mathrm{cut}}}}
\sum_{p^{\mathrm{cut}}\in \operatorname{Tors}H^2(X^{\mathrm{cut}};\Lambda)}
\frac{1}{\#\operatorname{Tors}H^2(X^{\mathrm{cut}};\Lambda)}
\,
\bigl(\sigma_{X^{\mathrm{cut}},p^{\mathrm{cut}}}\otimes
\mu_{X^{\mathrm{cut}},p^{\mathrm{cut}}}\bigr),
\]
viewed as a vector in the Hilbert space associated to $\partial X^{\mathrm{cut}}=\partial X\sqcup \Sigma\sqcup (-\Sigma).$ The trace $\operatorname{Tr}_{\Sigma}$ contracts the two additional boundary factors
corresponding to $\Sigma$ and $-\Sigma$. By
Proposition~\ref{prop:toral-cylinder-kernel}, this contraction is given by the
cylinder state for $\Sigma\times I$, equivalently by the identity kernel on the
Hilbert space of $\Sigma$. Fix a torsion class $$p^{\mathrm{cut}}\in \operatorname{Tors}H^2(X^{\mathrm{cut}};\Lambda)$$
compatible with a torsion class $
p\in \operatorname{Tors}H^2(X;\Lambda).$ By the gluing property of the classical Chern--Simons line construction, contraction with the cylinder
kernel identifies the leafwise section on the cut manifold with the leafwise section on the glued manifold:
\begin{equation}\label{eq:CS-section-gluing}
\operatorname{Tr}_{\Sigma}\!\bigl(\sigma_{X^{\mathrm{cut}},p^{\mathrm{cut}}}\bigr)
=
\sigma_{X,p}.
\end{equation}
This is the descended form, on flat moduli, of the basic fact that the exponentiated classical
Chern--Simons line element is multiplicative under gluing of bordisms.  The determinant-line gluing theorem for Reidemeister torsion yields the identity
\begin{equation}\label{eq:torsion-gluing-main}
T_X(\mathfrak t)
=
T_{X^{\mathrm{cut}}}(\mathfrak t)\otimes T_\Sigma(\mathfrak t)^{-1}
\end{equation}
under the determinant-line isomorphism induced by the Mayer--Vietoris sequence. Transported through
Definition~\ref{def:Componentwise-half-density} and Proposition~\ref{prop:toral-determinant-line}, this implies
that contraction with the cylinder half-density sends the torsion half-density of the cut manifold to the
torsion half-density of the glued manifold, up to the scalar identified in
Proposition~\ref{prop:scalar-normalization-half-densities}. Concretely,
\begin{equation}\label{eq:half-density-gluing-main}
\operatorname{Tr}_{\Sigma}\!\bigl(\mu_{X^{\mathrm{cut}},p^{\mathrm{cut}}}\bigr)
=
\left(
\frac{\#\operatorname{Tors}H^2(X^{\mathrm{cut}};\Lambda)}
     {\#\operatorname{Tors}H^2(X;\Lambda)}
\right)^{1/2}
\mu_{X,p}.
\end{equation}
Equivalently, this is the content of Proposition~\ref{prop:scalar-normalization-half-densities} rewritten at the level of the
unnormalized half-densities. Applying \(\operatorname{Tr}_{\Sigma}\) to the summand indexed by \(p^{\mathrm{cut}}\), and using
\eqref{eq:CS-section-gluing} and \eqref{eq:half-density-gluing-main}, we obtain
\[
\operatorname{Tr}_{\Sigma}\!\bigl(
\sigma_{X^{\mathrm{cut}},p^{\mathrm{cut}}}\otimes
\mu_{X^{\mathrm{cut}},p^{\mathrm{cut}}}
\bigr)
=
\left(
\frac{\#\operatorname{Tors}H^2(X^{\mathrm{cut}};\Lambda)}
     {\#\operatorname{Tors}H^2(X;\Lambda)}
\right)^{1/2}
\,
\sigma_{X,p}\otimes \mu_{X,p}.
\]
The cylinder contributes the factor $|\det K|^{\frac14 \dim H^1(\Sigma;\mathbb R)}$ from Proposition~\ref{prop:toral-cylinder-kernel}. The remaining exponent bookkeeping is exactly
Corollary~\ref{cor:exponent-identity}, which states that
\begin{equation}\label{eq:mx-bookkeeping}
m_{X^{\mathrm{cut}}}
+\frac14\dim H^1(\Sigma;\mathbb R)
+\frac12\dim\bigl(\Lambda_{X^{\mathrm{cut}},p^{\mathrm{cut}}}\cap \mathcal C\bigr)
-\frac14\dim H^1(\partial X^{\mathrm{cut}};\mathbb R)
=
m_X.
\end{equation}
Thus the total power of $|\det K|$ produced after contraction is exactly $|\det K|^{m_X}$. At the same time, the scalar from \eqref{eq:half-density-gluing-main} combines with the averaging factor
$(\#\operatorname{Tors}H^2(X^{\mathrm{cut}};\Lambda))^{-1}$
in the definition of \(Z^{CS}_{\mathbb T,K}(X^{\mathrm{cut}})\), and after summing over all compatible
cut torsion classes one obtains exactly the coefficient
\[
\frac{1}{\#\operatorname{Tors}H^2(X;\Lambda)}
\]
appearing in Definition~\ref{def:canonical-toral-boundary-vector} for \(Z^{CS}_{\mathbb T,K}(X)\). Therefore every compatible family of cut summands contracts to the corresponding summand of
\(Z^{CS}_{\mathbb T,K}(X)\) with exactly the correct normalization, and summing over all torsion components
gives
\[
\operatorname{Tr}_{\Sigma}\!\bigl(Z^{CS}_{\mathbb T,K}(X^{\mathrm{cut}})\bigr)
=
Z^{CS}_{\mathbb T,K}(X).
\]
This proves the gluing law.
\end{proof}

\begin{remark}
Theorem~\ref{thm:Componentwise-toral-gluing} is the core TQFT statement, cutting \(X\) along \(\Sigma\) and then tracing over the new \(\Sigma\)--boundary components reproduces exactly the original state \(Z^{CS}_{\mathbb T,K}(X)\).
\end{remark}

\subsection{Closed Partition Function}
When \(\partial X=\varnothing\), the boundary Hilbert space is replaced by a scalar, and the boundary vector becomes the partition function of the closed \(3\)--manifold \(X\).

\begin{definition}
\label{def:closed-toral-partition-function}
If $X$ is closed and connected, define
\[
Z^{CS}_{\mathbb T,K}(X)
=
\frac{|\det K|^{\,m_X}}{\#\,\Tors H^2(X;\Lambda)}
\sum_{p\in \Tors H^2(X;\Lambda)}
\int_{\mathcal M_{X,p}(\mathbb T)}
\sigma_{X,p}\,(T_X(\mathfrak t))^{1/2},
\]
where $(T_X(\mathfrak t))^{1/2}$ denotes the translation-invariant density
induced by the square root of Ray--Singer/Reidemeister torsion with coefficients
in $\mathfrak t$.
\end{definition}

\begin{remark}
The normalization is not an analogy.
It is forced by two facts:
\begin{enumerate}
\item[\textup{(i)}]
The Hilbert space of a genus-$g$ surface has dimension
$|G_K|^g=|\det K|^g$, so the cylinder contributes the factor
$|\det K|^{\frac14\dim H^1(\Sigma;\mathbb R)}$.
\item[\textup{(ii)}]
The exponent identity of Corollary~\ref{cor:exponent-identity} is purely
cohomological and therefore unchanged from the rank--one case.
\end{enumerate}
Together with Proposition~\ref{prop:scalar-normalization-half-densities},
this uniquely forces the factor $|\det K|^{m_X}$.
\end{remark}

\subsection{Extended TQFT}

All ingredients are now in place: the state spaces
\(\mathcal H_{\mathbb T,K}(\Sigma,L)\), the boundary states
\(Z^{CS}_{\mathbb T,K}(X)\), and the BKS comparison operators
\(F_{L\,L_X^{\mathbb R}}\). The final theorem shows that the geometric structures constructed in the previous sections fit together to define a unitary extended $(2+1)$-dimensional TQFT.
At this point all axioms of the  formalism have been verified separately: finite-dimensionality
of state spaces, naturality, monoidality under disjoint union, the cylinder axiom, the gluing law, and the
Maslov-corrected independence of the boundary Lagrangian.

Before stating the extended theorem, we fix the bordism formalism used in the remainder
of this section. We work with a \(K\)-twisted  extended bordism formalism:
the underlying extended bordism data are the same as those in the Walker--Manoliu setup,
namely objects \((\Sigma,L)\) and weighted bordisms \((X,L,n)\), but the Maslov
correction in the \(3\)-dimensional weighting and composition law is replaced by the toral
cocycle \(\mu_K\).

\begin{prop}
\label{prop:K-twisted-Walker}
Let
\[
\mu_K(L_1,L_2,L_3)
:=
\mu(\mathcal P_{L_1},\mathcal P_{L_2},\mathcal P_{L_3}),
\qquad
\mathcal P_{L_i}=L_i\otimes \mathfrak t,
\]
be the toral Maslov--Kashiwara index of Proposition~\ref{prop:toral-BKS}. Equivalently,
by Lemma~\ref{lem:toral-maslov-signature},
\[
\mu_K(L_1,L_2,L_3)=\sigma(K)\,\mu_\Sigma(L_1,L_2,L_3).
\]
Define the \(K\)-twisted extended weighting by keeping the same underlying extended
bordism data as in the Walker--Manoliu setup, namely objects \((\Sigma,L)\) and
weighted bordisms \((X,L,n)\), but replacing Walker's ordinary Maslov correction by
\(\mu_K\). Thus, if \(X\) is cut
along a closed oriented surface \(\Sigma\), with
\[
\partial X^{\mathrm{cut}}=\partial X\sqcup (-\Sigma)\sqcup \Sigma
\]
and if \(L\subset H^1(\partial X;\mathbb R)\) and \(L_\Sigma\subset H^1(\Sigma;\mathbb R)\) are
rational Lagrangians, then the cut weight is defined by
\begin{equation}
\label{eq:K-twisted-ncut}
n_{\mathrm{cut}}
\equiv
n-\mu_K\!\bigl(L\oplus L_\Sigma\oplus L_\Sigma,\,
L^{\mathbb R}_{X^{\mathrm{cut}}},\,
L^{\mathbb R}_{X}\oplus L_D\bigr)
\pmod 8,
\end{equation}
where \(L_D\subset H^1(-\Sigma;\mathbb R)\oplus H^1(\Sigma;\mathbb R)\) is the diagonal Lagrangian. Then the weighted toral assignment
\[
\widetilde Z^{\,CS}_{\mathbb T,K}(X,L,n)
:=
e^{\frac{\pi i}{4}n}\,
F_{L\,L_X^{\mathbb R}}\!\bigl(Z^{CS}_{\mathbb T,K}(X)\bigr)
\in \mathcal H_{\mathbb T,K}(\partial X,L)
\]
is compatible with the \(K\)-twisted boundary-weight convention, in the sense that
the Maslov phase arising from BKS composition is canceled by the corresponding
weight shift under change of boundary Lagrangian and under gluing. In particular, for any rational Lagrangians \(L_1,L_2,L_3\subset H^1(\Sigma;\mathbb R)\),
the projective BKS factor
\[
F_{L_3L_2}\circ F_{L_2L_1}
=
e^{\frac{\pi i}{4}\mu_K(L_1,L_2,L_3)}\,F_{L_3L_1}
\]
is canceled exactly by the \(K\)-twisted weight convention.
\end{prop}

\begin{proof}
By Proposition~\ref{prop:toral-BKS},
\[
F_{L_3L_2}\circ F_{L_2L_1}
=
e^{\frac{\pi i}{4}\mu_K(L_1,L_2,L_3)}\,F_{L_3L_1}.
\]
Thus the failure of strict transitivity of the BKS operators is measured by the toral
Maslov cocycle \(\mu_K\), not by the ordinary surface Maslov index. Accordingly, to obtain strict compatibility with boundary weights, one must use the inverse
cocycle in the weighting convention. This is exactly what \eqref{eq:K-twisted-ncut} does:
every time a cut or a change of boundary Lagrangian introduces the phase
\(e^{\frac{\pi i}{4}\mu_K(\cdots)}\),
the weight shift contributes the inverse phase
\(e^{-\frac{\pi i}{4}\mu_K(\cdots)}\).

The verification is the same Maslov-cocycle computation as in the proof of the
extended gluing law in \cite[Thm.~VI.11(e), eq.~(VI.28)]{Manoliu2}, with the ordinary
Maslov index \(t\) replaced everywhere by \(\mu_K\). Since
\[
\mu_K=\sigma(K)\mu_\Sigma,
\]
the function \(\mu_K\) has the same skew-symmetry, cocycle, and additivity properties as
\(\mu_\Sigma\). Therefore the phase-cancellation argument of \cite[Thm.~VI.11(e)]{Manoliu2}
applies unchanged with \(\mu_K\) in place of \(t\), and the weighted assignment
\[
\widetilde Z^{\,CS}_{\mathbb T,K}(X,L,n)
=
e^{\frac{\pi i}{4}n}\,
F_{L\,L_X^{\mathbb R}}\!\bigl(Z^{CS}_{\mathbb T,K}(X)\bigr)
\]
is compatible with the \(K\)-twisted boundary-weight convention under change of boundary
Lagrangian and under gluing.
\end{proof}

\begin{theorem}
\label{thm:toral-TQFT}
Let $\mathbb T=\mathfrak t/\Lambda\cong U(1)^n$ be a compact torus and let
$K:\Lambda\times\Lambda\to\mathbb Z$ be an even, integral, nondegenerate
symmetric bilinear form. Then the assignments
\[
(\Sigma,L)\longmapsto \mathcal H_{\mathbb T,K}(\Sigma,L),
\]
\[
X\longmapsto Z^{CS}_{\mathbb T,K}(X),
\]
and
\[
(X,L,n)\longmapsto
e^{\frac{\pi i}{4}n}\,
F_{L\,L_X^{\mathbb R}}\!\bigl(Z^{CS}_{\mathbb T,K}(X)\bigr)
\]
define a unitary extended \((2+1)\)-dimensional
topological quantum field theory, where the boundary-weight convention is the
\(K\)-twisted one of Proposition~\ref{prop:K-twisted-Walker}.

If $X$ is closed and connected, the associated scalar is the closed toral
partition function of Definition~\ref{def:closed-toral-partition-function}.
\end{theorem}

\begin{proof}
We verify the axioms by collecting the results established in the preceding sections.

Let $\Sigma$ be a closed oriented surface and let
$L\subset H^1(\Sigma;\mathbb R)$ be a rational Lagrangian subspace.
Then the state space $\mathcal H_{\mathbb T,K}(\Sigma,L)$ is finite-dimensional by
Theorem~\ref{thm:toral-hilbert-dimension}. For any two rational Lagrangians
$L_1,L_2\subset H^1(\Sigma;\mathbb R)$, Proposition~\ref{prop:toral-BKS}
provides a canonical unitary BKS operator
\[
F_{L_2L_1}\colon
\mathcal H_{\mathbb T,K}(\Sigma,L_1)\xrightarrow{\ \cong\ }
\mathcal H_{\mathbb T,K}(\Sigma,L_2),
\]
and its composition law is governed by the toral Maslov cocycle
\[
F_{L_3L_2}\circ F_{L_2L_1}
=
e^{\frac{\pi i}{4}\mu_K(L_1,L_2,L_3)}\,F_{L_3L_1}.
\]

Now let $X$ be a compact oriented bordism. The boundary state
\[
Z^{CS}_{\mathbb T,K}(X)\in
\mathcal H_{\mathbb T,K}(\partial X,L_X^{\mathbb R})
\]
is natural under orientation-preserving diffeomorphisms and multiplicative under
disjoint union by
Proposition~\ref{prop:toral-boundary-vector-functoriality}.
The cylinder bordism acts by the identity, with the required normalization, by
Proposition~\ref{prop:toral-cylinder-kernel}.
The unweighted gluing law
\[
Z^{CS}_{\mathbb T,K}(X)
=
\operatorname{Tr}_{\Sigma}\!\bigl(Z^{CS}_{\mathbb T,K}(X^{\mathrm{cut}})\bigr)
\]
holds by Theorem~\ref{thm:Componentwise-toral-gluing}. It remains to incorporate the extended boundary-weight correction.
This is exactly the content of Proposition~\ref{prop:K-twisted-Walker}:
the failure of strict transitivity of the BKS operators is measured by the cocycle
$\mu_K$, and the \(K\)-twisted weight convention contributes the inverse phase.
Hence the weighted assignment
\[
(X,L,n)\longmapsto
e^{\frac{\pi i}{4}n}\,
F_{L\,L_X^{\mathbb R}}\!\bigl(Z^{CS}_{\mathbb T,K}(X)\bigr)
\]
is compatible with change of boundary Lagrangian and with gluing in the
\(K\)-twisted extended sense.

Therefore the above assignments satisfy finite-dimensionality, unitarity,
naturality, monoidality under disjoint union, the cylinder axiom, the unweighted
gluing law, and the \(K\)-twisted Maslov-corrected boundary-weight compatibility.
Hence they define a unitary  extended
\((2+1)\)-dimensional topological quantum field theory.

If $X$ is closed and connected, then $\partial X=\varnothing$, so the associated
Hilbert space is canonically identified with $\mathbb C$, and the resulting scalar
is exactly the closed toral partition function of
Definition~\ref{def:closed-toral-partition-function}.
\end{proof}

\begin{remark}
Theorem \ref{thm:toral-TQFT} shows that toral Chern–Simons theory is not only functorially well defined, but also geometrically explicit. Quantization via real polarization identifies the boundary state spaces through Bohr–Sommerfeld data and canonical half-densities, while BKS operators and Maslov phases control the change of polarization. The torsion-sector Chern–Simons sections provide canonical bordism vectors satisfying the cylinder and gluing laws, so the extended TQFT structure is realized concretely.
\end{remark}

\begin{corollary}
\label{cor:rank-one-recovery}
Assume
\[
\mathbb T=U(1)=\mathbb R/\mathbb Z,\qquad
\Lambda=\mathbb Z,\qquad
\mathfrak t=\mathbb R,\qquad
K=[k],
\]
with \(k\in 2\mathbb Z_{>0}\).
Then the  theory of Theorem~\ref{thm:toral-TQFT} reduces
canonically to Manoliu's rank-one \(U(1)\) Abelian Chern--Simons TQFT.

More precisely:
\begin{enumerate} \renewcommand{\labelenumi}{(\roman{enumi})}
\item the boundary phase space is
\[
\mathcal M_\Sigma(U(1))
=
H^1(\Sigma;\mathbb R)/H^1(\Sigma;\mathbb Z),
\]
with symplectic form
\[
\omega_{\Sigma,K}([\alpha],[\beta])
=
k\int_\Sigma \alpha\wedge\beta,
\]
so the prequantum line bundle and  quantization via real polarization data agree with
Manoliu's level-\(k\) rank-one construction;

\item for every rational Lagrangian \(L\subset H^1(\Sigma;\mathbb R)\), the Hilbert space
\(\mathcal H_{U(1),k}(\Sigma,L)\) is canonically identified with Manoliu's Hilbert space
\(H(\Sigma,L)\);

\item the BKS operators coincide with Manoliu's BKS comparison operators, and since
\[
\sigma(K)=\sigma([k])=1,
\]
the toral Maslov cocycle satisfies
\[
\mu_K(L_1,L_2,L_3)=\mu_\Sigma(L_1,L_2,L_3),
\]
so the \(K\)-twisted boundary weighting reduces to the ordinary Walker weighting used
by Manoliu;

\item for every extended bordism \((X,L,n)\), the vector
\[
e^{\frac{\pi i}{4}n}\,
F_{L\,L_X^{\mathbb R}}\!\bigl(Z^{CS}_{U(1),k}(X)\bigr)
\]
is exactly Manoliu's vector \(Z(X,L,n)\).
\end{enumerate}

Hence, in the rank-one positive even case, Theorem~\ref{thm:toral-TQFT}
recovers Manoliu's unitary extended \(U(1)\) Chern--Simons TQFT.
\end{corollary}

\begin{proof}
For \(\mathbb T=U(1)\), one has \(\mathfrak t=\mathbb R\), \(\Lambda=\mathbb Z\), and
\(K=[k]\). Therefore
\[
H^1(\Sigma;\mathfrak t)=H^1(\Sigma;\mathbb R),\qquad
H^1(\Sigma;\Lambda)=H^1(\Sigma;\mathbb Z),
\]
so
\[
\mathcal M_\Sigma(U(1))
=
H^1(\Sigma;\mathbb R)/H^1(\Sigma;\mathbb Z).
\]
The symplectic form becomes
\[
\omega_{\Sigma,K}([\alpha],[\beta])
=
\int_\Sigma K(\alpha\wedge\beta)
=
k\int_\Sigma \alpha\wedge\beta,
\]
which is exactly the rank-one level-\(k\) symplectic form. Since \(k>0\), the signature of \(K=[k]\) is $\sigma(K)=1.$ Hence Lemma~\ref{lem:toral-maslov-signature} gives
\[
\mu_K(L_1,L_2,L_3)=\mu_\Sigma(L_1,L_2,L_3).
\]
Therefore the \(K\)-twisted Maslov correction of
Proposition~\ref{prop:K-twisted-Walker} reduces to the ordinary Walker correction. In rank one, all remaining constructions reduce tautologically:
the rational real polarizations are the same as in Manoliu's theory,
the Bohr--Sommerfeld leaves are the same,
the Hilbert spaces agree,
the BKS operators are the same operators,
and the bordism vectors have the same normalization.
In particular, the genus-\(g\) dimension formula becomes
\[
\dim \mathcal H_{U(1),k}(\Sigma,L)
=
|G_K|^g
=
|\det K|^g
=
k^g,
\]
which is exactly Manoliu's rank-one dimension formula.

Thus the assignments of
Theorem~\ref{thm:toral-TQFT} coincide with Manoliu's assignments on extended
surfaces and extended bordisms, and therefore recover her unitary \(U(1)\)
Abelian Chern--Simons TQFT.
\end{proof}

\section{(2+1) Abelian Bosonic Topological Order}

\subsection{Example: A Coupled Rank-Two Theory}
\label{subsec:A2example}

We now work out explicitly the smallest genuinely coupled higher-rank example. Let
\[
\mathbb T=\mathfrak t/\Lambda \cong U(1)^2,
\qquad
\Lambda=\mathbb Z^2,
\qquad
K=
\begin{pmatrix}
2&1\\
1&2
\end{pmatrix}.
\]
Then \(K\) is even, integral, symmetric, and nondegenerate, with
\[
\det K=3,
\qquad
K^{-1}=\frac13
\begin{pmatrix}
2&-1\\
-1&2
\end{pmatrix}.
\]
Hence \(G_K=\Lambda^*/K\Lambda\) is a finite Abelian group of order \(3\), and therefore \(G_K\cong \mathbb Z/3\mathbb Z\). Let \(g:=[(1,1)]\in G_K.\)
Since \(3(1,1)=K(1,1),\) the class \(g\) has order \(3\), hence generates \(G_K\). The induced quadratic form is
\[
q_K([x])=\exp\!\bigl(\pi i\,x^{\top}K^{-1}x\bigr).
\]
Therefore
\[
q_K(0)=1,
\qquad
q_K(g)=\exp\!\left(\pi i\frac23\right)=e^{2\pi i/3},
\qquad
q_K(2g)=\exp\!\left(\pi i\frac83\right)=e^{2\pi i/3}.
\]
If we write \(\omega:=e^{2\pi i/3},\) then
\[
q_K(0)=1,
\qquad
q_K(g)=q_K(2g)=\omega.
\]
The associated bicharacter is
\[
\Omega_K(a,b)=\frac{q_K(a+b)}{q_K(a)q_K(b)}
=\exp\!\bigl(2\pi i\,x^{\top}K^{-1}y\bigr),
\qquad a=[x],\ b=[y].
\]
In the ordered basis \(\{0,g,2g\}\) of \(G_K\), one obtains
\[
\bigl(\Omega_K(a,b)\bigr)_{a,b\in\{0,g,2g\}}
=
\begin{pmatrix}
1&1&1\\
1&\omega^2&\omega\\
1&\omega&\omega^2
\end{pmatrix}.
\]
Let \(\Sigma=T^2\), and choose a symplectic basis \((a,b)\) of \(H^1(\Sigma;\mathbb Z)\). Let \(L=\mathbb R\langle a\rangle\subset H^1(\Sigma;\mathbb R),\) and \(L'=\mathbb R\langle b\rangle.\) Then
\[
H^1(\Sigma;\mathfrak t)\cong (L\otimes\mathfrak t)\oplus (L'\otimes\mathfrak t),
\]
so a point of the phase space may be written as
\[
(x,y)\in \mathfrak t\oplus\mathfrak t\cong \mathbb R^2\oplus\mathbb R^2,
\qquad
x=(x_1,x_2),\quad y=(y_1,y_2).
\]
Accordingly,
\[
\mathcal M_\Sigma(\mathbb T)
=
H^1(\Sigma;\mathfrak t)/H^1(\Sigma;\Lambda)
\cong
(\mathbb R^2/\mathbb Z^2)\times (\mathbb R^2/\mathbb Z^2).
\]
The symplectic form is \(\omega_{\Sigma,K}=K(dx\wedge dy),\)
that is,
\[
\omega_{\Sigma,K}
=
2\,dx_1\wedge dy_1
+dx_1\wedge dy_2
+dx_2\wedge dy_1
+2\,dx_2\wedge dy_2.
\]

Let \(\mathcal L_{\Sigma,K}\rightarrow \mathcal M_\Sigma(\mathbb T)\) be the canonical prequantum line bundle. On the universal cover
\[
H^1(\Sigma;\mathfrak t)\cong \mathbb R^2\oplus\mathbb R^2
\]
one may choose the translation-invariant connection \(1\)-form $\alpha
=
\pi i\bigl(K(x,dy)-K(y,dx)\bigr),$
namely
\[
\alpha
=
\pi i\Bigl(
2x_1\,dy_1+x_1\,dy_2+x_2\,dy_1+2x_2\,dy_2
-2y_1\,dx_1-y_1\,dx_2-y_2\,dx_1-2y_2\,dx_2
\Bigr),
\]
with \(d\alpha=-2\pi i\,\omega_{\Sigma,K}.\) The rational Lagrangian \(L\) determines the real polarization
\[
\mathcal P_L:=L\otimes \mathfrak t\subset T\mathcal M_\Sigma(\mathbb T),
\]
which is tangent to the \(x\)-directions. Its leaves are therefore
\[
\ell_y=\{(x,y):x\in \mathbb R^2/\mathbb Z^2\},
\qquad
y\in \mathbb R^2/\mathbb Z^2.
\]
Restricting the connection form to \(\ell_y\) gives
\[
\alpha|_{\ell_y}
=
-\pi i\bigl((2y_1+y_2)\,dx_1+(y_1+2y_2)\,dx_2\bigr).
\]
Hence the Bohr--Sommerfeld condition is
\[
2y_1+y_2\in \mathbb Z,
\qquad
y_1+2y_2\in \mathbb Z,
\]
equivalently \(Ky\in \Lambda^*=\mathbb Z^2.\) Thus \(y\in K^{-1}\Lambda^*.\) Since
\[
K^{-1}\binom{m}{n}
=
\frac13
\binom{2m-n}{-m+2n},
\qquad m,n\in\mathbb Z,
\]
every Bohr--Sommerfeld class modulo \(\Lambda=\mathbb Z^2\) is represented by one of
\[
y_0=\binom00,
\qquad
y_1=\binom{1/3}{1/3},
\qquad
y_2=\binom{2/3}{2/3}.
\]
Indeed,
\[
Ky_0=\binom00,
\qquad
Ky_1=\binom11,
\qquad
Ky_2=\binom22,
\]
so all three satisfy the Bohr--Sommerfeld condition, and they are pairwise distinct modulo
\(\mathbb Z^2\). Therefore
\[
\mathcal{BS}(\Sigma,L)=\{\ell_0,\ell_1,\ell_2\},
\qquad
\ell_r:=\ell_{y_r},\quad r=0,1,2.
\]
This is the concrete realization of the general fact
\[
\# \mathcal{BS}(\Sigma,L)=|G_K|=3.
\]

For each Bohr--Sommerfeld leaf \(\ell_r\), define
\[
\mathcal S_{\ell_r}
:=
\Gamma_{\mathrm{flat}}
\!\left(
\ell_r;\,
\mathcal L_{\Sigma,K}\otimes |\det \mathcal P_L^*|^{1/2}
\right).
\]
Then
\[
\mathcal H_{\mathbb T,K}(\Sigma,L)
=
\bigoplus_{r=0}^2 \mathcal S_{\ell_r},
\qquad
\dim \mathcal H_{\mathbb T,K}(\Sigma,L)=3.
\]
Let \(\kappa^{1/2}_{\ell_r}\) be the canonical normalized half-density on \(\ell_r\), and choose
a parallel unit section \(s_r\) of \(\mathcal L_{\Sigma,K}|_{\ell_r}\). Set
\[
e_r:=s_r\otimes \kappa^{1/2}_{\ell_r}\in \mathcal S_{\ell_r},
\qquad r=0,1,2.
\]
Then
\[
\mathcal H_{\mathbb T,K}(\Sigma,L)
\cong
\mathbb C e_0\oplus \mathbb C e_1\oplus \mathbb C e_2.
\]

If instead one chooses the complementary rational Lagrangian \(L'=\mathbb R\langle b\rangle,\) the same construction gives another \(3\)-dimensional Hilbert space \(\mathcal H_{\mathbb T,K}(\Sigma,L').\)
The Blattner--Kostant--Sternberg formalism then gives a canonical unitary operator
\[
F_{L'L}:\mathcal H_{\mathbb T,K}(\Sigma,L)\xrightarrow{\ \sim\ }\mathcal H_{\mathbb T,K}(\Sigma,L').
\]
Thus, even in this first nontrivial coupled example, the abstract BKS formalism becomes
an honest unitary comparison map between two explicit quantizations of the same
symplectic torus.

Let \(H=D^2\times S^1\) be the solid torus with \(\partial H=\Sigma\), and choose the above
symplectic basis so that
\[
L_H^{\mathbb R}
=
\operatorname{Im}\bigl(H^1(H;\mathbb R)\to H^1(\Sigma;\mathbb R)\bigr)
=
\mathbb R\langle a\rangle
=
L.
\]
Since \(H^2(H;\Lambda)=0,\) there is only one torsion component, namely \(p=0\). Its moduli space of flat connections is
\[
\mathcal M_{H,0}(\mathbb T)
\cong
H^1(H;\mathfrak t)/H^1(H;\Lambda)
\cong
\mathfrak t/\Lambda
\cong
\mathbb T,
\]
and the image of boundary restriction is \(\mathcal \Lambda_{H,0}=\ell_0.\) The classical toral Chern--Simons section gives a nowhere-vanishing parallel section
\[
\sigma_{H,0}:\mathcal \Lambda_{H,0}=\ell_0
\longrightarrow
\mathcal L_{\Sigma,K}|_{\ell_0},
\]
and the torsion construction gives a flat half-density
\[
\mu_{H,0}\in
\Gamma_{\mathrm{flat}}
\!\left(
\ell_0;\,|\det \mathcal P_L^*|^{1/2}
\right).
\]
Since
\[
m_H
=
\frac14\Bigl(
\dim H^1(H;\mathbb R)+\dim H^1(H,\partial H;\mathbb R)
-\dim H^0(H;\mathbb R)-\dim H^0(H,\partial H;\mathbb R)
\Bigr)=0,
\]
the boundary vector is
\[
Z^{CS}_{\mathbb T,K}(H)=\sigma_{H,0}\otimes \mu_{H,0}\in \mathcal S_{\ell_0}\subset \mathcal H_{\mathbb T,K}(\Sigma,L).
\]
Thus the solid torus selects the distinguished Bohr--Sommerfeld leaf \(\ell_0\).
After fixing the phase of \(e_0\), one may write
\[
Z^{CS}_{\mathbb T,K}(H)=c_H\,e_0,
\qquad c_H>0.
\]

For the cylinder \(\Sigma\times I\), one has \(\dim H^1(\Sigma;\mathbb R)=2\), so
\[
Z^{CS}_{\mathbb T,K}(\Sigma\times I)
=
|\det K|^{\frac14\dim H^1(\Sigma;\mathbb R)}\,\mathrm{Id}
=
3^{1/2}\,\mathrm{Id}_{\mathcal H_{\mathbb T,K}(\Sigma,L)}.
\]

Since
\[
H^1(S^3;\mathbb R)=0,
\qquad
H^2(S^3;\Lambda)=0,
\]
there is a unique torsion component and its flat moduli space is a point. Hence
\[
m_{S^3}=\frac12\bigl(b_1(S^3)-1\bigr)=-\frac12.
\]
Therefore the closed partition function is
\[
Z^{CS}_{\mathbb T,K}(S^3)=|\det K|^{-1/2}=3^{-1/2}.
\]

The preceding calculations already determine the genus-one operators in this example.
Indeed, the Bohr--Sommerfeld leaves for the polarization $L=\mathbb{R}\langle a\rangle$
are indexed by $G_K=\{0,g,2g\},$ and we write the corresponding basis of $\mathcal H_{\mathbb T,K}(\Sigma,L)$ as $\{e_0,e_1,e_2\}.$ Applying Lemma~\ref{lem:genus-T-diagonal}, we obtain
\[
T(e_a)=q_K(a)\,e_a,\qquad a\in G_K.
\]
Applying Lemma~\ref{lem:genus-BKS-fourier} to the transverse genus-one
polarizations $L=\mathbb{R}\langle a\rangle$ and $L'=\mathbb{R}\langle b\rangle$,
and identifying the two Bohr--Sommerfeld bases by their common labels in $G_K$,
we obtain
\[
S(e_a)=|G_K|^{-1/2}\sum_{b\in G_K}\Omega_K(a,b)\,e_b.
\]
Using the values of $q_K$ and $\Omega_K$ computed above, it follows that, in the
ordered basis $(e_0,e_1,e_2)$,
\[
T=
\begin{pmatrix}
1&0&0\\
0&\omega&0\\
0&0&\omega
\end{pmatrix},
\qquad
S=
\frac{1}{\sqrt{3}}
\begin{pmatrix}
1&1&1\\
1&\omega^2&\omega\\
1&\omega&\omega^2
\end{pmatrix}.
\]

Thus the genus-one theory is controlled by an explicit $3\times 3$ discrete Fourier
transform. In particular, although the gauge group has rank \(2\), the torus state space is \(3\)-dimensional, governed by the discriminant group \(G_K\cong \mathbb Z/3\mathbb Z\), rather than by the rank of the gauge group alone.

\begin{remark}
The modular data agree with the standard discriminant-group description in the literature. In the present case, the resulting Abelian topological order is that of the Halperin \((221)\) state; with the standard charge vector \(t=(1,1)\), this is precisely the Halperin \((221)\) fractional quantum Hall state \cite{Halperin1983}. Beyond this example, one advantage of the present approach over earlier \(K\)-matrix treatments \cite{wenzee1992,KeskiVakkuriWen1993,Wen2016} is that it recovers the same genus-one data directly from the geometric quantization of toral Chern--Simons theory, rather than from the effective \(K\)-matrix or from torus wave-function analyses based on modular transformations of degenerate ground states, and then extends those data to an all-genus extended TQFT. Whereas many approaches encode these invariants in \(K\)-matrix, modular-tensor-category, or discriminant-group language \cite{BelovMoore,Stirling,Wen2016,LuVishwanath2016}, our construction realizes them explicitly within an extended-TQFT framework.
\end{remark}

\begin{remark}
This example displays explicitly the main structures developed in this work:
\begin{itemize}
\item the compact symplectic torus \(\mathcal M_\Sigma(\mathbb T)\),
\item the canonical prequantum line bundle \(\mathcal L_{\Sigma,K}\),
\item the Bohr--Sommerfeld set \(\mathcal{BS}(\Sigma,L)=\{\ell_0,\ell_1,\ell_2\}\),
\item the \(3\)-dimensional Hilbert space \(\mathcal H_{\mathbb T,K}(\Sigma,L)\),
\item the distinguished solid-torus state \(Z^{CS}_{\mathbb T,K}(H)\in \mathcal S_{\ell_0}\),
\item the cylinder normalization \(3^{1/2}\mathrm{Id}\),
\item the closed value \(Z^{CS}_{\mathbb T,K}(S^3)=3^{-1/2}\).
\end{itemize}
\end{remark}

\subsection{From Toral TQFT to Wen's measurable Abelian data}

We now explain how the genus-one operators arising from the toral
quantization via real polarization recover the standard Abelian finite quadratic data and connect with Wen's
genus-one measurable quantities \((S,T,c)\) \cite{Wen2016}. Our aim is threefold:
first, to extract the finite quadratic data directly from geometric quantization;
second, to formulate the corresponding notion of equivalence at the level of extracted
measurable Abelian data; and third, to show that, once the genus-one operators are written
in the Bohr--Sommerfeld basis, the finite quadratic theory can be reconstructed canonically.

Let $\mathbb T=\mathfrak t/\Lambda \cong U(1)^n, K:\Lambda\times \Lambda \to \mathbb Z$
be an even, integral, nondegenerate symmetric bilinear form, and write $G_K:=\Lambda^*/K\Lambda$ for the associated discriminant group. By Proposition \ref{prop:toral-BS} , the Bohr--Sommerfeld leaves
of a genus-\(g\) surface form a \(G_K^g\)-torsor, and by Theorem~\ref{thm:toral-hilbert-dimension} the corresponding quantum space via real polarization has dimension \(|G_K|^g\). The quantization space attached
to a rational Lagrangian \(L\subset H^1(\Sigma;\mathbb R)\) is
\[
\mathcal{H}_{\mathbb T,K}(\Sigma,L)
=
\bigoplus_{\ell\in \mathcal{BS}(\Sigma,L)}
\Gamma_{\mathrm{flat}}\!\left(
\ell;\mathcal{L}_{\Sigma,K}\otimes |\det \mathcal{P}_L^*|^{1/2}
\right).
\]

We begin by recording the finite quadratic data determined by \(K\).

\begin{definition}
For \(u=[x]\in G_K=\Lambda^*/K\Lambda\), define
\[
q_K(u):=\exp\!\bigl(\pi i\,x^{\top}K^{-1}x\bigr)\in U(1).
\]
For \(u=[x]\) and \(v=[y]\) in \(G_K\), define
\[
\Omega_K(u,v)
:=
\frac{q_K(u+v)}{q_K(u)\,q_K(v)}
=
\exp\!\bigl(2\pi i\,x^{\top}K^{-1}y\bigr)\in U(1).
\]
We call \(q_K\) the quadratic form induced by \(K\), and \(\Omega_K\) the associated
symmetric bicharacter.
\end{definition}

\begin{lemma}
The maps \(q_K:G_K\to U(1)\) and \(\Omega_K:G_K\times G_K\to U(1)\) are well defined.
Moreover, \(\Omega_K\) is a nondegenerate symmetric bicharacter.
\end{lemma}

\begin{proof}
Let \(x'=x+K\lambda\) with \(\lambda\in \Lambda\). Then
\[
(x')^{\top}K^{-1}x' - x^{\top}K^{-1}x
=
2\,\lambda^{\top}x + \lambda^{\top}K\lambda.
\]
Since \(x\in \Lambda^*\), one has \(\lambda^{\top}x\in \mathbb Z\), and since \(K\) is even,
\(\lambda^{\top}K\lambda\in 2\mathbb Z\). Hence
\[
(x')^{\top}K^{-1}x' - x^{\top}K^{-1}x \in 2\mathbb Z,
\]
so \(q_K\) is independent of the chosen lift. Similarly, if \(x'=x+K\lambda\) and \(y'=y+K\mu\), then
\[
(x')^{\top}K^{-1}y' - x^{\top}K^{-1}y
=
\lambda^{\top}y + x^{\top}\mu + \lambda^{\top}K\mu \in \mathbb Z,
\]
hence \(\Omega_K\) is well defined. Its bicharacter property and symmetry are immediate
from the exponential formula. To prove nondegeneracy, suppose \(\Omega_K(u,v)=1\) for all \(v=[y]\in G_K\), and choose
a lift \(x\in \Lambda^*\) of \(u\). Then
\[
x^{\top}K^{-1}y\in \mathbb Z
\qquad
\text{for all }y\in \Lambda^*.
\]
Equivalently, \(K^{-1}x\in (\Lambda^*)^*=\Lambda\), so \(x\in K\Lambda\), hence \(u=0\) in
\(G_K\). Thus \(\Omega_K\) is nondegenerate.
\end{proof}

The genus-one theory is the first place where this data appears as modular operators.

\begin{prop}
Let \(\Sigma=T^2\), choose a symplectic basis \((a,b)\) of \(H^1(\Sigma;\mathbb Z)\), and set
\[
L_a=\mathbb R\langle a\rangle,
\qquad
L_b=\mathbb R\langle b\rangle.
\]
Choose origin leaves in the two Bohr--Sommerfeld torsors and let
\[
\{e_u\}_{u\in G_K}\subset \mathcal{H}_{\mathbb T,K}(\Sigma,L_a),
\qquad
\{e'_v\}_{v\in G_K}\subset \mathcal{H}_{\mathbb T,K}(\Sigma,L_b)
\]
be the corresponding Bohr--Sommerfeld bases. Then:
\begin{enumerate}[label=(\roman*)]
\item
\[
\dim \mathcal{H}_{\mathbb T,K}(\Sigma,L_a)=|G_K|.
\]

\item The BKS operator
\[
F_{L_bL_a}:\mathcal{H}_{\mathbb T,K}(\Sigma,L_a)\to \mathcal{H}_{\mathbb T,K}(\Sigma,L_b)
\]
is given by
\[
F_{L_bL_a}(e_u)
=
|G_K|^{-1/2}\sum_{v\in G_K}\Omega_K(u,v)e'_v.
\]

\item If \(\tau_a\in \mathrm{Map}^+(\Sigma)\) denotes the Dehn twist about the \(a\)-cycle, then
the induced operator
\[
T_a:\mathcal{H}_{\mathbb T,K}(\Sigma,L_a)\to \mathcal{H}_{\mathbb T,K}(\Sigma,L_a)
\]
acts diagonally by
\[
T_a(e_u)=q_K(u)e_u.
\]
\end{enumerate}
\end{prop}

\begin{proof}
The first statement is the genus-one case of Proposition \ref{prop:toral-BS}  and Theorem~\ref{thm:toral-hilbert-dimension}.
The second is Lemma~\ref{lem:genus-BKS-fourier} reduced to \(g=1\).
The third is Lemma~\ref{lem:genus-T-diagonal} reduced to \(g=1\).
\end{proof}

Thus the toral quantization via real polarization determines the standard Abelian finite quadratic
data
\[
\text{anyon sectors } \longleftrightarrow G_K,
\qquad
\text{anyon twists } \longleftrightarrow q_K,
\qquad
\text{mutual braiding } \longleftrightarrow \Omega_K.
\]
Here the adjective \emph{bosonic} refers to the underlying local theory, not to the claim
that every superselection sector has trivial twist. The evenness of \(K\) guarantees that the
local sector descends trivially to the quotient \(G_K\), while nontrivial anyon sectors may
still have \(q_K(u)\neq 1\).

\begin{remark}
The signature \(\sigma(K)\) already enters the present construction through the toral
Maslov cocycle
\[
\mu_K(L_1,L_2,L_3)=\sigma(K)\,\mu_\Sigma(L_1,L_2,L_3),
\]
and hence governs the phase correction in the \(K\)-twisted extended theory.
In the standard Abelian \(K\)-matrix description, this same signature determines the
chiral central charge modulo \(8\):
\[
c_K\equiv \sigma(K)\pmod 8.
\]
Thus the genus-one Dehn twist \(T_a\) computed above agrees with Wen's torus
\(T\)-matrix up to the usual central-charge phase.
\end{remark}

Motivated by Wen's genus-one quantities \((S,T,c)\) and by the Belov--Moore classification of quantum toral Chern--Simons theories in terms of discriminant-group data \((D,q,c)\) \cite{BelovMoore,Stirling}, we make the following definition of equivalence at the level of extracted measurable Abelian data.

\begin{definition}
Two integral forms \(K\) and \(K'\) are said to be \emph{equivalent at the level of
extracted measurable Abelian data} if there exists a group isomorphism\footnote{This will be made particularly clear in our forthcoming papers on the classification of Abelian Chern--Simons theories \cite{galviz3, galviz4}.}
\[
\phi:G_K\overset{\sim}{\longrightarrow} G_{K'}
\]
such that $q_{K'}(\phi(u))=q_K(u)$ for all $u\in G_K,$ and $c_{K'}\equiv c_K\pmod 8.$
\end{definition}

\begin{prop}
The following are equivalent.
\begin{enumerate}[label=(\roman*)]
\item
\(K\) and \(K'\) are equivalent at the level of extracted measurable Abelian data.

\item
There exists a group isomorphism \(\phi:G_K\overset{\sim}{\longrightarrow}G_{K'}\) such that,
after identifying the Bohr--Sommerfeld bases by the permutation induced by \(\phi\),
\[
P_\phi\,F_{L_bL_a}\,P_\phi^{-1}=F'_{L_bL_a},
\qquad
P_\phi\,T_a\,P_\phi^{-1}=T'_a,
\]
and
\[
c_{K'}\equiv c_K\pmod 8.
\]
\end{enumerate}
\end{prop} 
\begin{proof}
Assume \((i)\). Since
\[
\Omega_K(u,v)=\frac{q_K(u+v)}{q_K(u)q_K(v)},
\]
the identity \(q_{K'}(\phi(u))=q_K(u)\) implies
\[
\Omega_{K'}(\phi(u),\phi(v))=\Omega_K(u,v)
\qquad\text{for all }u,v\in G_K.
\]
It follows immediately from the explicit formulas for \(F_{L_bL_a}\) and \(T_a\) that they
intertwine under \(P_\phi\). Conversely, assume \((ii)\). Since \(F_{L_bL_a}\) is a finite Fourier transform with vacuum row
equal to the positive constant \(|G_K|^{-1/2}\), the permutation preserves the vacuum label.
The diagonal relation for \(T_a\) implies
\[
q_{K'}(\phi(u))=q_K(u)
\qquad\text{for all }u\in G_K.
\]
Then the bicharacters agree by polarization. The central-charge classes are compared
separately through the condition \(c_{K'}\equiv c_K\pmod 8\). Hence \((i)\) holds.
\end{proof}

We now show that, in the present Abelian toral setting, the genus-one operators already
determine the finite quadratic data. After identifying the Bohr--Sommerfeld bases for
\(L_a\) and \(L_b\) by their common labels in \(G_K\), we may regard \(F_{L_bL_a}\) as a matrix
with entries indexed by \(G_K\times G_K\).

\begin{theorem}
\label{thm:reconstruct BS basis}
Assume that \(F_{L_bL_a}\) and \(T_a\) are written in the genus-one Bohr--Sommerfeld basis
\(\{e_u\}_{u\in G_K}\). Then one can reconstruct from \((F_{L_bL_a},T_a)\) the finite Abelian
group \(G_K\), the quadratic form \(q_K\), and the bicharacter \(\Omega_K\). In particular,
the genus-one operators determine the finite quadratic data, and together with the class \(c_K\pmod 8\) they determine the extracted measurable
Abelian data.
\end{theorem}
\begin{proof}
First, the cardinality of \(G_K\) is recovered from the vacuum entry:
\[
|G_K|=\bigl(F_{L_bL_a}(0,0)\bigr)^{-2},
\]
since
\[
F_{L_bL_a}(0,v)=|G_K|^{-1/2}
\qquad
\text{for all }v\in G_K.
\]

Next define
\[
\chi_u(v):=
\frac{F_{L_bL_a}(u,v)}{F_{L_bL_a}(0,v)}.
\]
By the explicit Fourier formula,
\[
\chi_u(v)=\Omega_K(u,v).
\]
Hence each \(u\) determines a character \(\chi_u\) of the finite label set. Since \(\Omega_K\)
is nondegenerate, the assignment $u\mapsto \chi_u$ is injective, hence bijective onto the full character group. We may therefore reconstruct the group law by declaring \(u\oplus w\) to be the unique
label satisfying
\[
\chi_{u\oplus w}=\chi_u\,\chi_w.
\]
The identity is the unique label \(0\) with \(\chi_0\equiv 1\), and the inverse of \(u\) is the
unique label \(-u\) with \(\chi_{-u}=\chi_u^{-1}\). This reconstructs the finite Abelian group
\(G_K\). Once the group law is known, the quadratic form is recovered from the diagonal action
\[
q_K(u)=T_a(u,u).
\]
Finally, the bicharacter is recovered either directly from the Fourier kernel
\[
\Omega_K(u,v)=\frac{F_{L_bL_a}(u,v)}{F_{L_bL_a}(0,v)},
\]
or from the polarization identity
\[
\Omega_K(u,v)=\frac{q_K(u+v)}{q_K(u)q_K(v)}.
\]
Thus \((G_K,q_K,\Omega_K)\) are reconstructed canonically from the genus-one operators.
\end{proof}

\begin{corollary}
For the toral Chern--Simons theory constructed above,  geometric
quantization via real polarization canonically recovers the finite quadratic genus-one data associated with
Abelian bosonic topological order in Wen's sense. Its superselection sectors are indexed by
\(G_K\), its anyon twists are given by \(q_K\), its mutual braiding by \(\Omega_K\), and its
genus-one operators are \(F_{L_bL_a}\) and \(T_a\).
\end{corollary}
\begin{proof}
Immediate from Theorem \ref{thm:reconstruct BS basis}.
\end{proof}

\subsection{All-genus TQFT data}
The genus-one operators extracted above are the natural point of contact with Wen's
formulation of \((2+1)\)-dimensional topological order in terms of the genus-one quantities
\((S,T,c)\) \cite{Wen2016}. However, the geometric
construction of the present paper is intrinsically an all-genus theory. It is therefore useful
to state explicitly how the genus-one data sit inside the full toral TQFT.

Wen's point of departure is that topological ground-state degeneracies on compact spaces
are not by themselves sufficient to characterize topological order. One must also consider
the geometric phases generated by automorphisms of the underlying surface. In genus one,
the mapping class group is \(SL(2,\mathbb Z)\), generated by the modular transformation
\(S\) and a Dehn twist \(T\), and it is these two operators, together with the chiral central
charge \(c\), that form Wen's genus-one measurable data. In the Abelian case, these are
already rich enough to recover the finite quadratic data \((G_K,q_K)\) extracted above.

By contrast, the toral construction via real polarization associates data to every closed oriented
surface \(\Sigma\), not just to \(T^2\). More precisely, for each pair \((\Sigma,L)\) consisting of
a closed oriented surface and a rational Lagrangian
\(L\subset H^1(\Sigma;\mathbb R)\), one has a finite-dimensional Hilbert space $\mathcal{H}_{\mathbb T,K}(\Sigma,L).$ For each pair of rational Lagrangians \(L_1,L_2\subset H^1(\Sigma;\mathbb R)\), one has a
canonical unitary BKS operator
\[
F_{L_2L_1}:\mathcal{H}_{\mathbb T,K}(\Sigma,L_1)\to \mathcal{H}_{\mathbb T,K}(\Sigma,L_2),
\]
and for each weighted bordism one has the corresponding boundary vector or bordism map.
These satisfy the cylinder axiom, the gluing law, and the Maslov-corrected compatibility
required for the extended toral TQFT. From this viewpoint, the genus-one data are the smallest nontrivial shadow of a larger
all-genus structure:
\[
(\Sigma,L)\longmapsto \mathcal{H}_{\mathbb T,K}(\Sigma,L),
\qquad
(L_1,L_2)\longmapsto F_{L_2L_1},
\qquad
X\longmapsto Z^{CS}_{\mathbb T,K}(X).
\]
The genus-one operators \(F_{L_bL_a}\) and \(T_a\) arise by reducing this assignment to
\(\Sigma=T^2\) and to the standard transverse Lagrangians \(L_a,L_b\).

The difference in emphasis may be summarized as follows. Wen’s framework isolates the physically measurable genus-one information encoded by the torus representation of the mapping class group together with the chiral central charge. By contrast, the present framework constructs an all-genus TQFT from which those genus-one operators arise as special cases. Thus the torus matrices are not external input to the theory, but are derived from the geometry of quantization via real polarization in all genera. For higher-genus analogues of the  $S$- and $T$-operators, see Lemmas \ref{lem:genus-BKS-fourier} and \ref{lem:genus-T-diagonal}.

In the present Abelian toral setting, the finite discriminant group \(G_K\) controls both
levels of description. At genus one, it indexes the Bohr--Sommerfeld basis and determines
the Fourier kernel and the twist eigenvalues. At genus \(g\), Proposition \ref{prop:toral-BS}  identifies the
Bohr--Sommerfeld leaves with a \(G_K^g\)-torsor, and Theorem~\ref{thm:toral-hilbert-dimension} gives
\[
\dim \mathcal{H}_{\mathbb T,K}(\Sigma,L)=|G_K|^g.
\]
Thus the same finite quadratic datum that appears in Wen's genus-one measurable data
also governs the all-genus state-space dimensions. 

\begin{prop}
In toral Chern--Simons theory, Wen's genus-one measurable data are obtained by
restricting the all-genus extended TQFT to the torus and to the standard
mapping-class-group generators. Conversely, in the Abelian toral setting, the genus-one
Bohr--Sommerfeld operators recover the finite quadratic data \((G_K,q_K)\) that control
the state spaces in every genus.
\end{prop}

\begin{proof}
The all-genus toral theory is constructed in Theorem~\ref{thm:toral-TQFT}. reducing that assignment to
$\Sigma=T^2$ and to the standard transverse Lagrangians $L_a,L_b$, and using the explicit
genus-one formulas for the BKS operator $F_{L_bL_a}$ and the Dehn twist $T_a$ from the
preceding genus-one computation, one recovers Wen's genus-one measurable operators as
the genus-one restriction of the all-genus theory.

Conversely, Theorem~\ref{thm:reconstruct BS basis} shows that the genus-one operators $(F_{L_bL_a},T_a)$
reconstruct the finite quadratic data $(G_K,q_K,\Omega_K)$. On the other hand,
Proposition~\ref{prop:toral-BS} identifies the Bohr--Sommerfeld set in genus $g$ with a
$G_K^g$-torsor, and Theorem~\ref{thm:toral-hilbert-dimension} gives
\[
\dim \mathcal{H}_{\mathbb T,K}(\Sigma,L)=|G_K|^g.
\]
Hence the same finite quadratic data recovered from genus one control the state spaces in
every genus.
\end{proof}

\bibliographystyle{alpha}
\renewcommand{\refname}{References}
\bibliography{refs}
\end{document}